%% file: ms.tex
\documentclass[journal,twoside]{IEEEtran}
\usepackage{amsthm}
\usepackage{amsfonts}
\usepackage{amsmath}
\usepackage{mathrsfs}
\usepackage{accents}
\usepackage{color}
\usepackage{rotating}
\usepackage{multirow}
\usepackage{hyperref}

\theoremstyle{plain}
\newtheorem{theorem}{Theorem}
\newtheorem{proposition}{Proposition}
\newtheorem{lemma}{Lemma}
\newtheorem{assumption}{Assumption}
\newtheorem{corollary}{Corollary}

\theoremstyle{definition}
\newtheorem{examplex}{Example}
\newenvironment{example}
{\pushQED{\qed}\examplex}
{\popQED\endexamplex}

\newcommand \ubar[1]{\underaccent{\bar}{#1}}
\newcommand \at[2]{_{#1 \mid #2}}

\begin{document}
\title{Warm Start of Mixed-Integer Programs \\ for Model Predictive Control of Hybrid Systems}

\author{Tobia~Marcucci and Russ~Tedrake% <-this % stops a space
\thanks{T.~Marcucci  and R.~Tedrake are with the Computer Science and Artificial Intelligence Laboratory (CSAIL),
	Massachusetts Institute of Technology,
	Cambridge, MA 02139, USA.
	E-mail: \texttt{\{tobiam, russt\}@mit.edu}.
	}% <-this % stops a space
%\thanks{Manuscript received MONTH DAY, YEAR; revised MONTH DAY, YEAR.}
}

\maketitle

\begin{abstract}
\input{sections/abstract.tex}
\end{abstract}

\section{Introduction}
\label{sec:intro}
\input{sections/intro.tex}

\section{Hybrid Model Predictive Control}
\label{sec:background}
\input{sections/background.tex}

\section{Hybrid MPC via Branch-and-Bound}
\label{sec:bb}
\input{sections/bb.tex}

\section{Construction of the Initial Cover}
\label{sec:initial_cover}
\input{sections/initial_cover.tex}

\section{Propagation of Subproblem Lower Bounds}
\label{sec:lower_bounds}
\input{sections/lower_bounds.tex}

\section{Propagation of an Upper Bound}
\label{sec:upper_bound}
\input{sections/upper_bound.tex}

\section{Asymptotic Analysis}
\label{sec:asymptotic_analysis}
\input{sections/asymptotic_analysis.tex}

\section{Terminal Penalties and Constraints}
\label{sec:terminal}
\input{sections/lower_bounds_terminal.tex}

\section{Numerical Study}
\label{sec:numerical_study}
\input{sections/numerical_study.tex}

\section{Conclusions}
\label{sec:conclusions}
\input{sections/conclusions.tex}

\appendices

\section{Extensions and Additional Applications}
\label{sec:extensions}
\input{sections/extensions.tex}

\section{Lagrangian Dual of the Convex Relaxation of~\eqref{eq:miqp}}
\label{sec:dual_qp}
\input{sections/dual_qp.tex}

\section{Proof of Theorem~\ref{th:propagation_bounds}}
\label{sec:proof_theorem}
\input{sections/proof_theorem.tex}

\section*{Acknowledgment}
\input{sections/acknowledgment.tex}

\ifCLASSOPTIONcaptionsoff
  \newpage
\fi


\bibliographystyle{IEEEtran}
\bibliography{bibliography}


%\newpage
\begin{IEEEbiography}[{\includegraphics[width=1in,height=1.25in,clip,keepaspectratio]{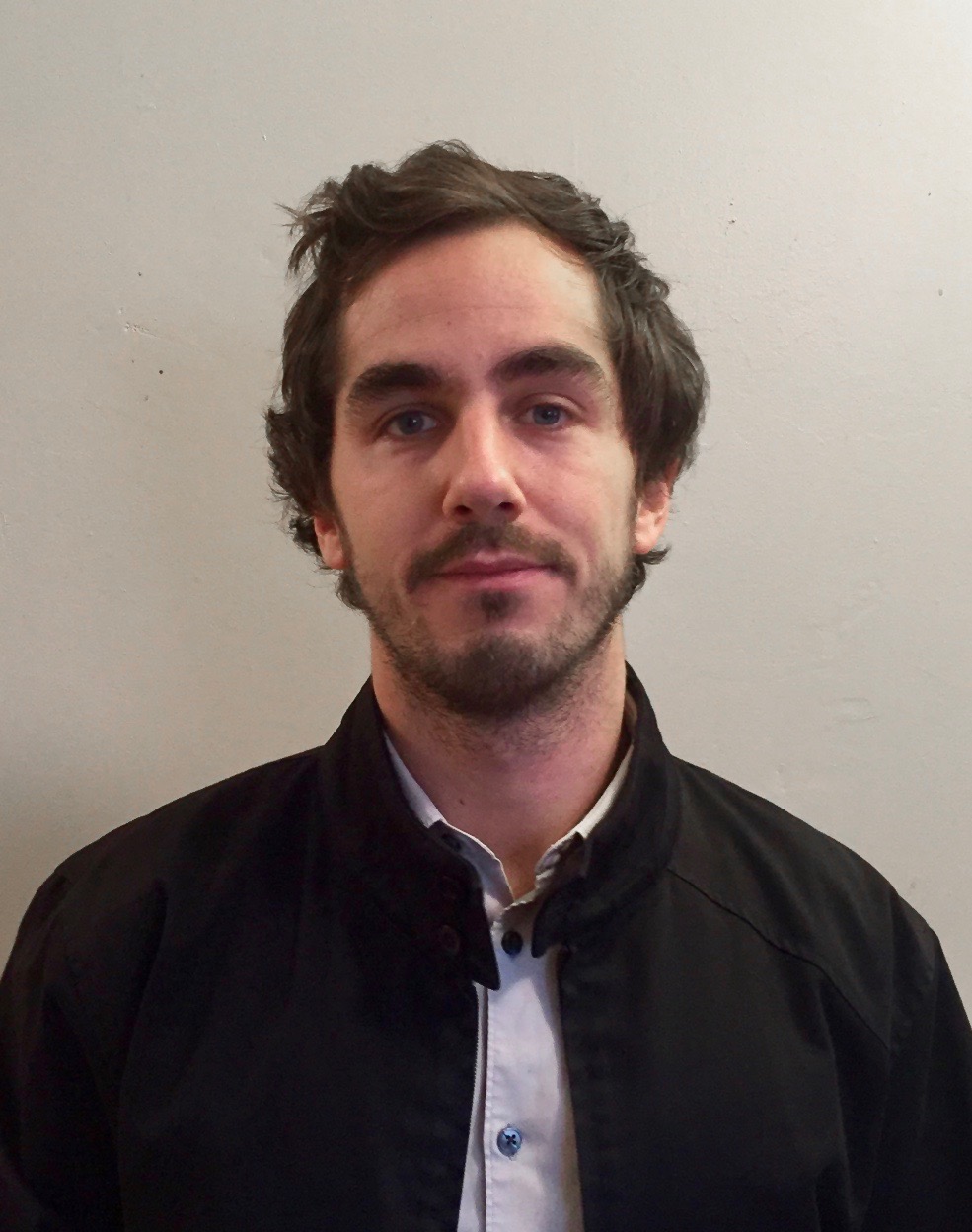}}]{Tobia Marcucci}
\input{sections/tobia_bio.tex}
\end{IEEEbiography}
\begin{IEEEbiography}[{\includegraphics[width=1in,height=1.25in,clip,keepaspectratio]{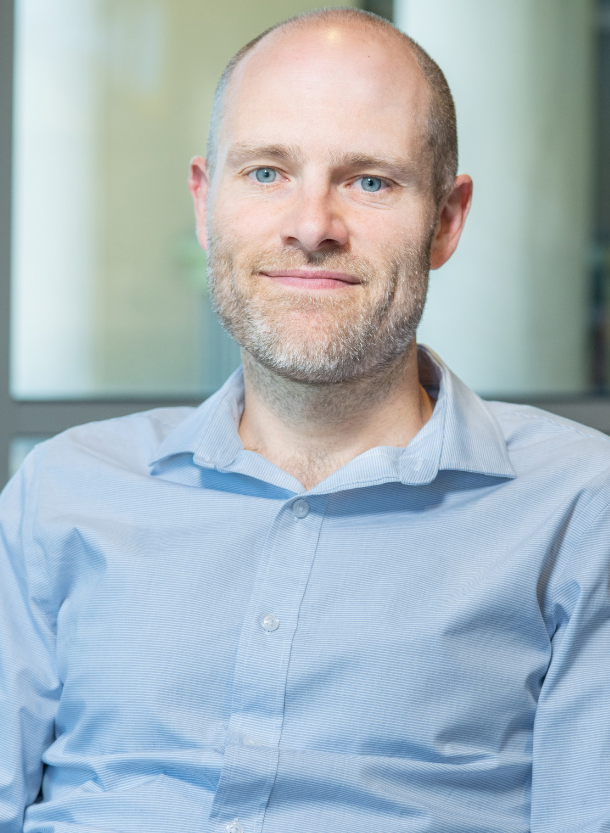}}]{Russ Tedrake}
\input{sections/russ_bio.tex}
\end{IEEEbiography}
%\vfill

%\newpage



\end{document}

%% file: sections/abstract.tex
In hybrid Model Predictive Control (MPC), a Mixed-Integer Quadratic Program (MIQP) is solved at each sampling time to compute the optimal control action.
Although these optimizations are generally very demanding, in MPC we expect consecutive problem instances to be nearly identical.
This paper addresses the question of how computations performed at one time step can be reused to accelerate (\emph{warm start}) the solution of subsequent MIQPs.

Reoptimization is not a rare practice in integer programming:
for small variations of certain problem data, the branch-and-bound algorithm allows an efficient reuse of its search tree and the dual bounds of its leaf nodes.
In this paper we extend these ideas to the receding-horizon settings of MPC.
The warm-start algorithm we propose copes naturally with arbitrary model errors, has a negligible computational cost, and frequently enables an a-priori pruning of most of the search space.
Theoretical considerations and experimental evidence show that the proposed method tends to reduce the combinatorial complexity of the hybrid MPC problem to that of a one-step look-ahead optimization, greatly easing the online computation burden.

%% file: sections/intro.tex
\IEEEPARstart{M}{odel} Predictive Control (MPC) is a numerical technique that enables the design of optimal feedback controllers for a wide variety of dynamical systems~\cite{mayne2000constrained,bemporad1999control}.
The main idea behind it is straightforward: if we are able to solve trajectory optimization problems quickly enough, we can replan the future motion of the system at each sampling time and achieve a reactive behavior.
While for smooth dynamics the online computations of MPC are generally limited to a simple convex program (even in the nonlinear case~\cite{diehl2005real}), the discrete behavior of hybrid systems is most naturally modeled with integer variables, requiring the real-time solution of mixed-integer programs.
This can be prohibitive even for systems with ``slow dynamics'' and of ``moderate size.''

The focus of this paper is \emph{hybrid linear systems}, i.e., systems whose nonlinearity is exclusively due to discrete logics.
For these, in the common case of a quadratic cost, the MPC problem falls in the class of Mixed-Integer Quadratic Programs (MIQPs).
MIQPs are NP-hard problems and, as such, no polynomial-time algorithm is known for their solution.
The most robust and effective strategy for tackling this class of optimizations is Branch and Bound (B\&B)~\cite{conforti2014integer,fletcher1998numerical}.
Despite its worst-case performance, this algorithm is very appealing: for a feasible optimization, B\&B converges to a global optimum; otherwise, it provides a certificate of infeasibility.

B\&B solves an MIQP by constructing a search tree, where at each node a Quadratic Program (QP) is solved to bound the objective function over a subset of the search space.
As an order of magnitude, for large-scale control problems, B\&B can easily require millions of QPs to converge~\cite{marcucci2019mixed}.
It is therefore natural to ask whether at the end of the time step all the information contained in the tree is necessarily lost, or it can be reused to warm start the solution of the next MIQP.
This seems plausible considering that two consecutive optimizations overlap for most of the time horizon, and differ only for a one-step shift of the time window.
This idea has been extremely successful in linear MPC (see, e.g.,~\cite{rao1998application,ferreau2008online,kuindersma2014efficiently,stellato2020osqp,liao2020fbstab}), but its application in the hybrid case raises many difficulties and has been obstructed by the complexity of B\&B algorithms.

\subsection{Related Works}
\input{sections/related_works.tex}

\subsection{Contribution}
\input{sections/contribution.tex}

\subsection{Article Organization}
\input{sections/organization.tex}

\subsection{Notation}
\input{sections/notation.tex}

%% file: sections/related_works.tex
Given the difficulty of solving MIQPs online, techniques to compute offline the optimal control as a function of the system state have been intensively developed~\cite{dua2002multiparametric,borrelli2005dynamic,oberdieck2015explicit}, and also extended beyond hybrid linear systems~\cite{gueddar2012approximate,charitopoulos2016explicit}.
However, the application of these \emph{explicit} methods is typically limited to low-dimensional systems, with very few discrete variables.
Approximate explicit solutions to the hybrid MPC problem have been proposed in~\cite{axehill2014parametric,marcucci2017approximate,sadraddini2019sampling}.
These extend the scope of exact approaches, but still require a substantial amount of offline computations, which might not be feasible in many applications.
In fact, certain problem data might be known only at run time, excluding the possibility of solving the MPC problem offline.

Noticing that the hardness of these problems lies in the identification of the optimal integer assignment, one can devise a split of the problem into two: a cheap algorithm to generate a good guess for the integers, followed by a rounding step~\cite{sager2009direct,sager2011combinatorial,jung2015lagrangian,frick2015embedded}.
This is a popular approach for hybrid nonlinear systems, and warm starting is having a crucial role in its advancement~\cite{burger2019design}.
However, it is not particularly convenient in our context, since the rounding step above typically suffers from the same combinatorics as our original MIQP.

Even though heuristic~\cite{takapoui2017simple} and local~\cite{frick2019low} methods have recently been proved to be very effective, B\&B is still the most reliable algorithm for solving hybrid MPC problems online~\cite[Section~17.4]{borrelli2017predictive}.
Many enhancements of B\&B tailored to MPC have been proposed, and attention has been mainly focused on accelerating the solution of the quadratic subproblems.
To this end, various algorithms have been considered: dual active set~\cite{axehill2006mixed}, dual gradient projection~\cite{axehill2008dual,naik2017embedded}, interior point~\cite{frick2015embedded}, partially-nonnegative least squares~\cite{bemporad2015solving,bemporad2018numerically}, and alternating-direction method of multipliers~\cite{stellato2018embedded}.
Search heuristics that leverage the problem temporal structure have also been proposed~\cite{bemporad1999efficient,frick2015embedded}.

Most of the B\&B schemes mentioned above make full use of warm start within a single B\&B solve, using the parent solution as a starting point for the child subproblems.
However, the issue of reusing computations across time steps has only been discussed in~\cite{bemporad2018numerically}.
There, a guess of the optimal integer assignment (obtained by shifting the previous solution) is prioritized within the construction of the new tree.
A similar approach has recently been proposed in~\cite{hespanhol2019structure}, where the whole path from the B\&B root to the optimal leaf is propagated in time.
Even if these techniques can lead to considerable savings, limiting the data propagated across time steps to a guess of the optimal solution is generally very restrictive.
In fact, in practice, we often expect disturbances to make these guesses inaccurate.
More importantly, even in the ideal case in which the integer warm start is actually optimal, these methods still build a B\&B tree almost from scratch, requiring the solution of many subproblems.
Note that, for an MIQP, proving optimality of a candidate solution is in principle as hard as solving the original optimization~\cite{del2017mixed}.

The problem of warm starting (or reoptimizing) a Mixed-Integer Linear Program (MILP) is not new to the operations research community~\cite{ausiello2011complexity}.
For a sequence of MILPs with common constraint matrix, the general approach is to start each B\&B search from the final frontier (B\&B leaves) of the previous solver run~\cite{hiller2013reoptimization,gamrath2015reoptimization}.
Moreover, in case of changes in the constraint right-hand side only, the dual bases of the previous frontier can be used to bound the optimal values of the new leaves~\cite{ralphs2006duality}.
This is a much more comprehensive reuse of computations than what is currently done in MPC.
First, not only the optimal solution, but the whole B\&B tree is propagated between subsequent problems.
This is very important, since, before convergence, the B\&B algorithm might also need to thoroughly explore regions of the search space that are far away from the optimum.
Second, by propagating dual bounds between consecutive MILPs, these approaches are capable of pruning large branches of the tree without solving any subproblems.

The latter ideas do not transfer smoothly to MPC.
In the general case of a time-varying system, consecutive MPC problems do not share the same constraint matrix, and the techniques mentioned above do not apply.
In the time-invariant case, on the other hand, we could interpret a sequence of MPC problems as MIQPs with variable constraint right-hand side, as done in explicit MPC~\cite{bemporad2002explicit,borrelli2005dynamic}.
However, proceeding as in~\cite{ralphs2006duality}, B\&B solutions would be reused without being shifted in time, completely ignoring the receding-horizon structure of the problem at hand.

%% file: sections/contribution.tex
We present a novel warm-start procedure for hybrid MPC, which bridges the gap with state-of-the-art reoptimization techniques from operations research.
First, we show how an initial search frontier for the hybrid MPC problem can be obtained by shifting in time part of the final frontier of the previous B\&B tree.
Then, duality is used to derive tight bounds on the cost of the new subproblems.
Starting from this refined partition of the search space, B\&B generally requires only a few subproblems to find the optimum.
Then, the implied dual bounds readily prune most of the search space, accelerating convergence without sacrificing global optimality.
Neither the shift of the B\&B frontier, nor the synthesis of the bounds, causes any significant time overhead in the MIQP solves.

The proposed method copes naturally with model errors and disturbances of any magnitude.
Remarkably, as the time horizon grows, and the MPC policy becomes stationary, our approach reduces the hybrid MPC combinatorics to that of a one-step look-ahead problem.
In this asymptotic case, previous computations are fully reused and only the variables of the final time step have to be reoptimized.

We evaluate the performance of our algorithm with a thorough statistical analysis.
In the vast majority of the cases, it leads to a drastic reduction of computation times and, even in the worst case, it still performs better than the customary approach of solving each MIQP from scratch.

%% file: sections/organization.tex
We structured this paper trying to maximize readability.
We start in Section~\ref{sec:background} presenting a minimal formulation of the MPC problem, which contains only the components necessary to the development of the warm-start algorithm.
Section~\ref{sec:bb} reviews the B\&B algorithm, emphasizing the advantages of dual methods in the solution of the subproblems.
In the same section, we identify the three main ingredients that compose a warm start for an MIQP.
Sections~\ref{sec:initial_cover},~\ref{sec:lower_bounds}, and~\ref{sec:upper_bound} are devoted to showing how each of these ingredients can be efficiently computed for the minimal MPC problem at hand.
Section~\ref{sec:asymptotic_analysis} presents an asymptotic analysis of the algorithm as the MPC time horizon tends to infinity.
In Section~\ref{sec:terminal} we generalize the problem formulation from Section~\ref{sec:background}, and we extend the results to these more general settings.
A statistical study of the algorithm performance is reported in Section~\ref{sec:numerical_study}.
Section~\ref{sec:conclusions} is dedicated to conclusions.
In Appendix~\ref{sec:extensions} several extensions of the proposed warm-start method are discussed, whereas Appendices~\ref{sec:dual_qp} and~\ref{sec:proof_theorem} contain mathematical derivations.

%% file: sections/notation.tex
We denote the set of real numbers as $\mathbb R$ and, e.g., nonnegative reals as $\mathbb R_{\geq 0}$.
The same notation is used for integers $\mathbb Z$, and we let $\mathbb N := \mathbb Z_{\geq 0}$.
The Euclidean length of a vector $x \in \mathbb R^n$ is $|x|$.
We use the same symbol for the cardinality $| \mathcal S |$ of a set $\mathcal S$.
For two vectors $x \in \mathbb R^{n}$ and $y \in \mathbb R^{m}$, $(x, y) \in \mathbb R^{n + m}$ represents their concatenation.
For a matrix $A \in \mathbb R^{n \times m}$, we let $A'$ be its transpose, $A^+$ its pseudoinverse, $\| A \|$ its maximum singular value, and $\ker (A)$ its nullspace.
All physical units may be assumed to be expressed in the MKS system.

%% file: sections/background.tex
Many equivalent descriptions of hybrid linear systems can be found in the literature~\cite{heemels2001equivalence}, in this paper we employ the popular framework of Mixed Logical Dynamical (MLD) systems~\cite{bemporad1999control}.
This description naturally lends itself to mixed-integer optimization, and it is the intermediate representation in which hybrid systems are more commonly cast for numerical optimal control~\cite[Section~17.4]{borrelli2017predictive}.
In this section, we introduce MLD systems (Section~\ref{sec:mld}) and we formulate the associated MPC problem (Section~\ref{sec:hmpc}).

\subsection{Mixed Logical Dynamical Systems}
\label{sec:mld}
\input{sections/mld.tex}

\subsection{The Optimal Control Problem}
\label{sec:hmpc}
\input{sections/hmpc.tex}

%% file: sections/mld.tex
We compactly represent an MLD system as
\begin{align}
\label{eq:mld}
x_{\tau+1} = A x_\tau + B u_\tau + e_\tau, \quad (x_\tau, u_\tau) \in \mathcal D,
\end{align}
where
$x_\tau \in \mathbb R^{n_x}$ denotes the system state at discrete time $\tau \in \mathbb N$,
$u_\tau \in \mathbb R^{n_u} \times \{0,1\}^{m_u}$ collects continuous and binary inputs,
$e_\tau \in \mathbb R^{n_x}$ represents the model error,
and the domain $\mathcal D := \{ (x,u) \mid F x + G u \leq h \}$ is a polyhedral subset of $\mathbb R^{n_x + n_u + m_u}$ which contains the origin.
We denote by $v_\tau \in \{0,1\}^{m_u}$ the binary entries in the input vector, and we let $V$ be the selection matrix such that $v_\tau = V u_\tau$.

Even if the MLD model~\eqref{eq:mld} is more compact than the usual description employed in the MPC literature, it can be used without loss of generality:
\begin{itemize}
\item
Often a distinction between independent and dependent (auxiliary) input variables $u_\tau$ is made~\cite{bemporad1999control}.
For a \emph{well-posed} MLD system, the second are assumed to be uniquely determined by the first and the state $x_\tau$ through the constraint set $\mathcal D$.
However, the role of these variables is identical from the optimization viewpoint, so we do not distinguish between them here.
\item 
Affine MLD dynamics as in~\cite[Section~16.5]{borrelli2017predictive} can be made linear through a shift of the system coordinates around an equilibrium point $(\hat x, \hat u)$, provided that binary inputs are defined so that $V \hat u = 0$.
\item 
Binary states can be handled introducing auxiliary inputs.
Let $A_i$ and $B_i$ be the $i$th rows of $A$ and $B$, respectively, and let $u_{\tau,j} \in \{0,1\}$ be the $j$th input.
Enforcing $A_i x_\tau+ B_i u_\tau = u_{\tau,j}$ through the constraint set $\mathcal D$, we obtain $x_{\tau+1,i} \in \{0,1\}$.
\end{itemize}

Handling binary states with auxiliary inputs simplifies the analysis but can be computationally inefficient: Appendix~\ref{sec:binary_mld} shows how our warm-start algorithm can cope directly with binary states.
Moreover, paying the price of a heavier notation, the proposed method also applies to time-varying MLD systems.
This extension is presented in Appendix~\ref{sec:time_varying_mld}.

%% file: sections/hmpc.tex
We now describe the optimization problem beneath the hybrid MPC controller.
To streamline the presentation, in the main body of this paper, we consider the simplified problem statement~\eqref{eq:miqp} given below.
This formulation does not allow terminal penalties and constraints, which are fundamental tools to ensure the stability of the closed-loop system~\cite{mayne2000constrained}.
In Section~\ref{sec:terminal} we extend our algorithm to incorporate these important MPC ingredients.
Additionally, in this paper, we limit our attention to quadratic objective functions, even though the results we present can be easily adjusted in case of different convex costs (e.g., 1-norm or $\infty$-norm).

Under the assumption of a perfect model ($e_\tau = 0$ for all $\tau$), an MPC controller regulates system~\eqref{eq:mld} to the origin by solving an open-loop optimal control problem at each time step.
Let $\tau$ be the time step at which the optimization problem is solved (the \emph{current time}), and let $t \in \mathbb N$ denote the \emph{relative time} within the MPC problem.
Given the current state $x_\tau$, we formulate the MIQP
\begin{subequations}
\label{eq:miqp}
\begin{align}
\label{eq:miqp_objective}
\min \
&
\sum_{t = 0}^{T} | Q x\at{t}{\tau} |^2 +
\sum_{t = 0}^{T-1} | R u\at{t}{\tau} |^2\\
\label{eq:initial_conditions}
\mathrm{s.t.} \ & x\at{0}{\tau} = x_\tau, \\
\label{eq:miqp_dynamics}
& x\at{t+1}{\tau} = A x\at{t}{\tau} + B u\at{t}{\tau}, &&  t = 0, \ldots,  T-1, \\
\label{eq:miqp_constraints}
& (x\at{t}{\tau}, u\at{t}{\tau}) \in \mathcal D, && t = 0, \ldots,  T-1, \\
\label{eq:miqp_binaries}
& V u\at{t}{\tau} \in \{0,1\}^{m_u}, && t = 0, \ldots,  T-1.
\end{align}
\end{subequations}
Here the optimization variables are $\{ u\at{t}{\tau} \}_{t=0}^{T-1}$ and $\{x\at{t}{\tau} \}_{t=0}^{T}$, and the time horizon $T$ is assumed to be fixed (the case with variable horizon is briefly discussed in Appendix~\ref{sec:variable_horizon}).
We do not assume the objective~\eqref{eq:miqp_objective} to be strictly convex, i.e., the weight matrices $Q$ and $R$ are allowed to be rank deficient.

The outcome of~\eqref{eq:miqp} is an optimal (up to a tolerance $\varepsilon \in \mathbb R_{\geq 0}$) open-loop  control sequence $\{u\at{t}{\tau}^*\}_{t = 0 }^{T-1} $, with the related state trajectory $\{x\at{t}{\tau}^*\}_{t = 0}^{T}$.
In MPC, only the first action $u_{\tau} := u\at{0}{\tau}^*$ is applied to the system.
Then, at time step $\tau + 1$, the new current state $x_{\tau+1}$ is measured and problem~\eqref{eq:miqp} is solved in a receding-horizon fashion.
Given the similarity of the problems we solve at time $\tau$ and $\tau+1$, it is natural to ask whether part of the computations performed at one time step can be exploited to speed up the solution of the consecutive problem.
In the next section we introduce the notions necessary to formalize this question.

%% file: sections/bb.tex
This section reviews the bases of B\&B by considering its application to problem~\eqref{eq:miqp}.
In Section~\ref{sec:bb_algorithm}, we describe the main steps of the algorithm.
Placing a special emphasis on the input-output behavior of each iteration, we provide a simple formalization of the warm-start problem.
In Section~\ref{sec:duality_subproblem}, we discuss how Lagrangian duality can facilitate the solution of the B\&B subproblems.
For a more thorough description of B\&B, we refer the reader to, e.g.,~\cite[Section~9.2]{conforti2014integer}.

\subsection{The Branch-and-Bound Algorithm}
\label{sec:bb_algorithm}
\input{sections/bb_algorithm.tex}

\subsection{Lagrangian Duality in the Solution of the Subproblem}
\label{sec:duality_subproblem}
\input{sections/duality_subproblem.tex}

%% file: sections/bb_algorithm.tex
Generally, B\&B is presented as a tree search, where each node corresponds to a convex relaxation of the MIQP.
Here we emphasize the set-cover interpretation of B\&B, which enables a more fluent analysis of the warm-start problem.
Similar presentations can also be found in~\cite{hiller2013reoptimization,gamrath2015reoptimization}.

We denote problem~\eqref{eq:miqp} by $\mathbf P$ and its optimal value by $\theta \in \mathbb R_{\geq 0} \cup \{\infty\}$, where $\theta = \infty$ in case of an infeasible MIQP.
In this section, for simplicity, we do not explicitly annotate the dependence of problem $\mathbf P$ on the time step $\tau$.
The B\&B search relies on the solution of convex relaxations (or subproblems) of $\mathbf P$, where the nonconvex constraints~\eqref{eq:miqp_binaries} are replaced by the linear inequalities
\begin{align}
\label{eq:relaxed_binaries}
\ubar v\at{t}{\tau} \leq v\at{t}{\tau} := V u\at{t}{\tau} \leq \bar v\at{t}{\tau},
\end{align}
for some $\ubar v\at{t}{\tau}, \bar v\at{t}{\tau} \in \{0,1\}^{m_u}$ such that $\ubar v\at{t}{\tau} \leq \bar v\at{t}{\tau}$.
A convex relaxation of $\mathbf P$ is hence a QP identified by the interval
\begin{align}
\label{eq:interval_tau}
\mathcal V := [(\ubar v\at{t}{\tau})_{t=0}^{T-1}, (\bar v\at{t}{\tau})_{t=0}^{T-1}] \subset \mathbb R^{T m_u},
\end{align}
and we denote it by $\mathbf P (\mathcal V)$.
Similarly, $\theta (\mathcal V) \in \mathbb R_{\geq 0} \cup \{\infty\}$ will represent its optimal value.

At iteration $i \in \mathbb N$ of the B\&B algorithm, we are given three inputs:
\begin{enumerate}
\item
A collection $\mathscr V^i$ of intervals of the form~\eqref{eq:interval_tau}, whose union covers the set $\{0,1\}^{T m_u}$.
Each interval $\mathcal V$ in $\mathscr V^i$ determines a subproblem $\mathbf P (\mathcal V)$ which, in the tree interpretation of the algorithm, is a leaf node.
Analogously, the cover $\mathscr V^i$ can be understood as the whole B\&B frontier.
It is important to remark that we do not assume the tree to have a single root, i.e., we allow $| \mathscr V^0 | \geq 1$.
Without loss of generality, we can assume the sets in $\mathscr V^i$ to be disjoint.
\item 
A lower bound $\ubar \theta (\mathcal V) \in \mathbb R_{\geq 0} \cup \{\infty\}$ on the optimal value $\theta (\mathcal V)$ for each set  $\mathcal V$ in $\mathscr V^i$.
Except for root nodes, this represents the dual bound implied by the solution of the parent subproblem.
\item 
An upper bound $\bar \theta^i \in \mathbb R_{\geq 0} \cup \{\infty \}$ on the optimal value of $\mathbf P$.
This is the objective of the best (lowest in cost) subproblem solved so far that is binary feasible, i.e., whose solution verifies~\eqref{eq:miqp_binaries}.
\end{enumerate}

Central to this work is the choice of the B\&B inputs: the initial cover $\mathscr V^0$, the lower bounds $\ubar \theta (\mathcal V)$ for each $\mathcal V$ in $\mathscr V^0$, and the upper bound $\bar \theta^0$.
Clearly, in case no information about the solution is available, the initialization $\mathscr V^0 := \{ \mathcal V \}$, with $\mathcal V :=  [ 0,  1]^{T m_u}$ the unit hypercube, $\ubar \theta(\mathcal V ) := 0$, and $\bar \theta^0 := \infty$ is always valid.
On the other hand, as we will see in the following sections, the structure of problem~\eqref{eq:miqp} allows the synthesis of nontrivial B\&B initializations, leveraging the solutions coming from the previous time steps.

The $i$th iteration of B\&B consists of the following steps.
Given an optimality tolerance $\varepsilon$, we select a subproblem, identified by the set $\mathcal V^i \in \mathscr V^i$, such that
\begin{align}
\label{eq:convergence_bb}
\ubar \theta (\mathcal V^i ) < \bar \theta^i - \varepsilon.
\end{align}
We solve the convex program $\mathbf P (\mathcal V^i )$, and we apply the first valid condition from the following list:
\begin{enumerate}
\item
\emph{Pruning.}
If $\theta(\mathcal V^i ) \geq \bar \theta^i - \varepsilon$, any binary assignment in $\mathcal V^i$ cannot be ``$\varepsilon$-cheaper'' than the one we already have.
Hence, we set
$\ubar \theta(\mathcal V^i ) \leftarrow \theta(\mathcal V^i )$
and we let
$\mathscr V^{i+1} \leftarrow \mathscr V^i$,
$\bar \theta^{i+1} \leftarrow  \bar \theta^i$.
\item 
\emph{Solution update.}
If the condition for 1) is not met, and the solution of $\mathbf P(\mathcal V^i )$ is binary feasible, then the optimal value $\theta (\mathcal V^i )$ is an upper bound for the objective of $\mathbf P$, tighter than the one we have.
Hence we update the bounds
$\bar \theta^{i+1} \leftarrow \theta(\mathcal V^i )$ and
$\ubar \theta(\mathcal V^i ) \leftarrow \theta(\mathcal V^i )$, but we do not refine the cover
$\mathscr V^{i+1} \leftarrow \mathscr V^i$.

\item
\emph{Branching.} 
If neither 1) nor 2) applies, we select a time $t$ and an element of $v\at{t}{\tau}$ whose optimal value is not binary.
We then split $\mathcal V^i$ into two subsets, $\mathcal U^i$ and $\mathcal W^i$: one in which this element is forced to be zero, the other in which it equals one.
We then update the cover $\mathscr V^{i+1} \leftarrow \{ \mathcal U^i, \mathcal W^i \} \cup \mathscr V^i \backslash \{ \mathcal V^i \}$, and we leave the upper bound unchanged $\bar \theta^{i+1} \leftarrow  \bar \theta^i$.
The lower bounds $\ubar \theta (\mathcal U^i)$ and $\ubar \theta (\mathcal W^i)$ are obtained through a simple duality argument discussed in Section~\ref{sec:duality_subproblem}.
\end{enumerate}

The algorithm terminates when condition~\eqref{eq:convergence_bb} is not met for any set in $\mathscr V^i$, and returns the cover $\mathscr V^* := \mathscr V^i$ and the cost $\theta^* := \bar \theta^i \leq \theta + \varepsilon$.
Clearly, B\&B is a finite algorithm, since, in the worst case, it amounts to the enumeration of all the $2^{T m_u}$ potential binary assignments.

%% file: sections/duality_subproblem.tex
\begin{figure*}[t]
\normalsize
\begin{subequations}
\label{eq:dual}
\begin{align}
\label{eq:dual_objective}
\max \ & - \sum_{t = 0}^{T} | \rho\at{t}{\tau} / 2|^2 - \sum_{t = 0}^{T-1} ( | \sigma\at{t}{\tau} / 2 |^2 + h' \mu\at{t}{\tau} + \bar v\at{t}{\tau}' \bar \nu\at{t}{\tau} - \ubar v\at{t}{\tau}' \ubar \nu\at{t}{\tau} ) - x_\tau' \lambda\at{0}{\tau} \\
\mathrm{s.t.} \ &
\label{eq:dual_x_t}
Q' \rho\at{t}{\tau} + \lambda\at{t}{\tau} - A' \lambda\at{t+1}{\tau} + F' \mu\at{t}{\tau} = 0, && t = 0, \ldots, T-1, \\
\label{eq:dual_x_T}
& Q' \rho\at{T}{\tau} + \lambda\at{T}{\tau} = 0, \\
\label{eq:dual_u_t}
& R' \sigma\at{t}{\tau} - B' \lambda\at{t+1}{\tau} + G' \mu\at{t}{\tau} + V' (\bar \nu\at{t}{\tau} - \ubar \nu\at{t}{\tau}) = 0, && t = 0, \ldots, T-1, \\
\label{eq:dual_nonneg}
& (\mu\at{t}{\tau}, \ubar \nu\at{t}{\tau}, \bar \nu\at{t}{\tau}) \geq 0, && t = 0, \ldots, T-1.
\end{align}
\end{subequations}
\hrulefill
\end{figure*}

The algorithm we present in this paper makes use of the dual $\mathbf D (\mathcal V)$ of the subproblem $\mathbf P (\mathcal V)$.
However, this does not entail any practical limitation: most efficient B\&B implementations employ dual methods for the solution of the subproblems (see, e.g.,~\cite{fletcher1998numerical,buchheim2016feasible,axehill2006mixed,axehill2008dual,naik2017embedded}).
In this subsection, we analyze the structure of $\mathbf D (\mathcal V)$ and we briefly discuss the main affinities between Lagrangian duality and B\&B.

The dual $\mathbf D  (\mathcal V)$ is derived in Appendix~\ref{sec:dual_qp}, and reported in Equation~\eqref{eq:dual}.
Its decision variables are the following Lagrange multipliers:
\begin{itemize}
\item$\lambda\at{t}{\tau}$ associated, for $t=0$, with the initial conditions~\eqref{eq:initial_conditions} and, for $t \geq 1$, with the MLD dynamics~\eqref{eq:miqp_dynamics};
\item $\mu\at{t}{\tau}$ corresponding to the MLD constraints~\eqref{eq:miqp_constraints};
\item $\ubar \nu\at{t}{\tau}$ and $\bar \nu\at{t}{\tau}$ coupled with the lower and upper bounds~\eqref{eq:relaxed_binaries} on the relaxed binary variables;
\item $\rho\at{t}{\tau}$ and $\sigma\at{t}{\tau}$ resulting from the introduction of auxiliary primal variables needed to handle the rank deficiency of $Q$ and $R$ (see Appendix~\ref{sec:dual_qp}).
\end{itemize}
By strong duality, the optimal value of $\mathbf D  (\mathcal V)$ coincides with $\theta (\mathcal V)$.

The first thing we notice when analyzing $\mathbf D  (\mathcal V)$ is that all the B\&B subproblems share the same dual feasible set, since the primal bounds $\ubar v\at{t}{\tau}$ and $\bar v\at{t}{\tau}$ become cost coefficients in~\eqref{eq:dual}.
This allows us to use the dual solution of a subproblem both to warm start the child QPs and to find lower bounds on their optimal values.
The bounds $\ubar \theta (\mathcal U^i)$, $\ubar \theta (\mathcal W^i)$ required in the branching step can, in fact, be obtained simply by substituting the parent multipliers into the child objectives.
Note that, by nonnegativity of $\ubar \nu\at{t}{\tau}$, $\bar \nu\at{t}{\tau}$ and since descending in the B\&B tree the bounds $\ubar v\at{t}{\tau}$, $\bar v\at{t}{\tau}$ can only be tightened, we have $\ubar \theta (\mathcal U^i) \geq \ubar \theta (\mathcal V^i)$ and $\ubar \theta (\mathcal W^i) \geq \ubar \theta (\mathcal V^i)$.

Another advantage of working on the dual emerges during pruning.
Algorithms such as dual active set or dual gradient projection, which take great advantage of warm starts, converge to the optimal value $\theta  (\mathcal V^i)$ from below.
This allows us to prematurely terminate a QP solve whenever the threshold $\bar \theta^i - \varepsilon$ is exceeded, leading to considerable computational savings.

Finally, we observe that $\mathbf D  (\mathcal V)$ is always feasible, since setting all the multipliers to zero satisfies the constraints in~\eqref{eq:dual}.
This implies that unboundedness of the dual is not only sufficient but also necessary for infeasibility of the primal.
Therefore, when solving a primal-infeasible QP, a dual solver will detect a set of feasible multipliers whose cost $\ubar \theta (\mathcal V)$ is strictly positive and for which $\rho\at{t}{\tau} = 0$ and $\sigma\at{t}{\tau} =0$ for all $t$.
In fact, these dual variables can be scaled by an arbitrary positive coefficient while preserving feasibility and increasing the dual objective.
In the following, we will refer to such a set of multipliers as a \emph{certificate of infeasibility} for $\mathbf P (\mathcal V)$.

%% file: sections/initial_cover.tex
In Section~\ref{sec:bb}, we have seen that a warm start for problem~\eqref{eq:miqp} should consist of: an initial cover $\mathscr V^0$, a set of lower bounds $\ubar \theta (\mathcal V)$ for each set $\mathcal V$ in $\mathscr V^0$, and an upper bound $\bar \theta^0$ on the MIQP objective.
We now show how to efficiently construct these elements by leveraging the structure of problem~\eqref{eq:miqp}.
In this section, we focus on the initial cover $\mathscr V^0$.
Sections~\ref{sec:lower_bounds} and~\ref{sec:upper_bound} will be devoted to the synthesis of the lower bounds $\ubar \theta (\mathcal V)$  and the upper bound $\bar \theta^0$.
An illustrative example of the following procedure is given at the end of this section (see also Figure~\ref{fig:example}).

In the following, to distinguish between instances of problem~\eqref{eq:miqp} associated with different time steps, we make use of the subscript $\tau$.
For example, the MIQP~\eqref{eq:miqp} will be denoted by $\mathbf P_\tau$ and its initial cover by $\mathscr V_\tau^0$.
Without loss of generality, we consider the current time to be $\tau=1$.
We assume the previous optimization, $\mathbf P_0$, to be feasible, and we let $\mathscr V_0^*$ be the cover of $\{0,1\}^{T m_u}$ that we obtain from its solution.
By construction, $\mathscr V_0^*$ is composed of disjoint intervals $\mathcal V_0$ of the form~\eqref{eq:interval_tau}, i.e., $\mathcal V_0 := [(\ubar v\at{t}{0})_{t=0}^{T-1}, (\bar v\at{t}{0})_{t=0}^{T-1}]$.

We assemble the initial cover $\mathscr V_1^0$ as follows:
\begin{enumerate}
\item
Since at time $\tau=1$ the binary input $v_0$ applied to the system at $\tau=0$ is known, we discard from $\mathscr V_0^*$ all the intervals which do not agree with this control action.
More precisely, we only keep the sets $\mathcal V_0$ which satisfy the condition
\begin{align}
\label{eq:drop_interval_from_cover}
\ubar v\at{0}{0} \leq v_0 \leq \bar v\at{0}{0}.
\end{align}
\item 
For all the retained sets, we add to $\mathscr V_1^0$ the interval
\begin{multline}
\label{eq:shifted_interval}
\mathcal V_1 :=
[
(\ubar v\at{1}{0}, \ldots, \ubar v\at{T-1}{0} , \overbrace{0, \ldots, 0}^{m_u\text{ times}}), \\
(\bar v\at{1}{0}, \ldots, \bar v\at{T-1}{0}, \underbrace{1, \ldots, 1}_{m_u\text{ times}})
].
\end{multline}
In words, this operation shifts the bounds defining $\mathcal V_0$ one step backwards in time, and appends the trivial bound $[0,1]^{m_u}$ on the binaries of the new terminal stage.
\end{enumerate}
We now verify that the resulting collection of sets is a valid initialization for the B\&B algorithm.
\begin{proposition}
The collection $\mathscr V_1^0$ covers $\{0,1\}^{T m_u}$ and is composed of disjoint intervals.
\end{proposition}
\begin{proof}
Let $(v\at{t}{1} )_{t=0}^{T-1}$ be a generic element of $ \{0,1\}^{T m_u}$.
Since $\mathscr V_0^*$ covers $\{0,1\}^{T m_u}$, there must be a set in it that contains $( v_0, v\at{0}{1}, \ldots, v\at{T-2}{1} )$.
This implies, by construction, the existence of a set in $\mathscr V_1^0$ that contains $(v\at{t}{1} )_{t=0}^{T-1}$.
Hence $\mathscr V_1^0$ covers $\{0,1\}^{T m_u}$.
Now consider $(v\at{t}{1} )_{t=0}^{T-1} \in \mathbb R^{T m_u}$, and assume the existence of two sets in $\mathscr V_1^0$ which contain this point.
Then there must also be two sets in $\mathscr V_0^*$ which contain $( v_0, v\at{0}{1}, \ldots, v\at{T-2}{1} )$.
This contradicts our assumption on $\mathscr V_0^*$, hence the sets in $\mathscr V_1^0$ are disjoint.
\end{proof}

It should be noted that this shifting process propagates the whole B\&B frontier from one time step to the next, and not just the optimal solution as previously done in~\cite{bemporad2018numerically,hespanhol2019structure}.
As we analyze in depth in Section~\ref{sec:asymptotic_analysis}, this ensures that both the work done to identify the optimal solution and that necessary to prove its $\varepsilon$-optimality (which generally is the dominant computation effort) are reused across time steps.
We highlight that this construction can be entirely completed before the measurement of the next state $x_1$, hence it is not cause of any delay in the solution of the MIQP $\mathbf P_1$.

We conclude this section with a simple synthetic example, illustrated in Figure~\ref{fig:example}, of the procedure presented above.

\begin{example}
\label{ex:initial_cover_example}
\input{sections/initial_cover_example.tex}
\end{example}

%% file: sections/initial_cover_example.tex
\begin{figure*}[h!]
\centering
\begin{tabular}{c|c|c|c|}
\cline{2-4}
&
\multirow{2}{*}{Initial cover $\mathscr V_\tau^0$}
&
\multirow{2}{*}{B\&B tree}
&
\multirow{2}{*}{Final cover $\mathscr V_\tau^*$}
\\ &&& \\
\hline
\multicolumn{1}{|c|}{\begin{sideways}  Time $\tau=0$ \end{sideways}}
&
\fcolorbox{white}{white}{\includegraphics[height=.2\textwidth]{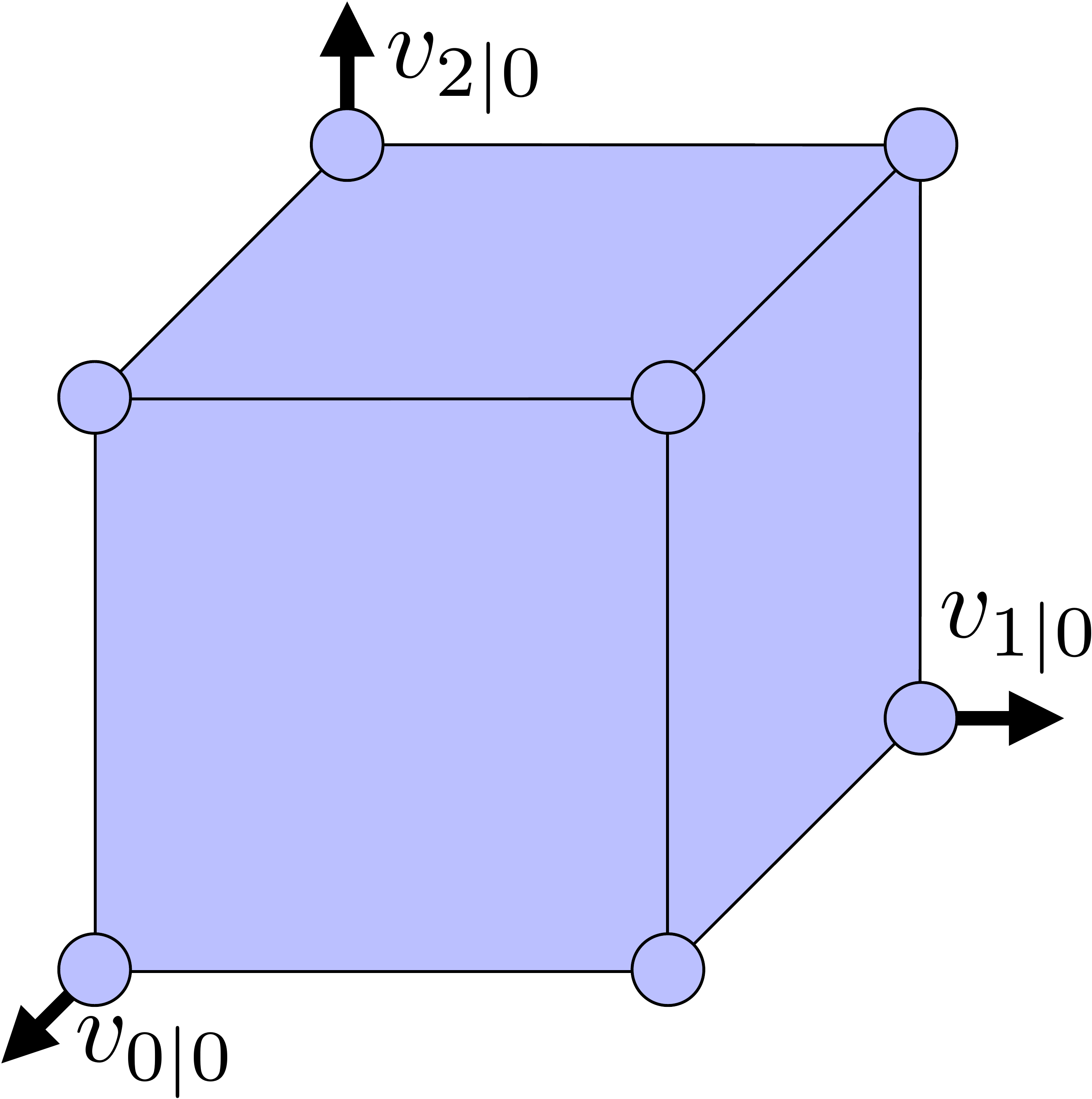}}
&
\fcolorbox{white}{white}{ \includegraphics[height=.2\textwidth]{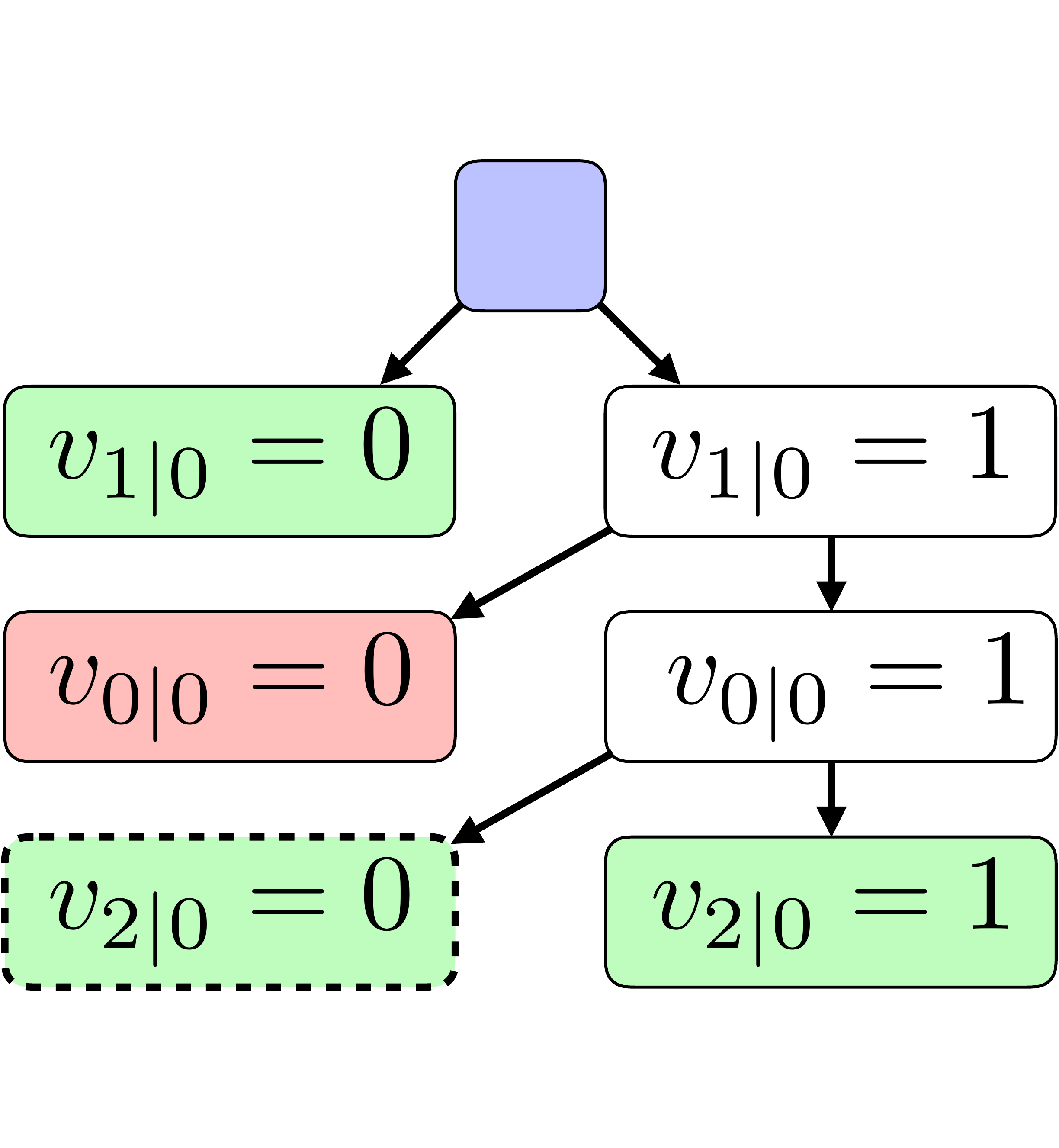}}
&
\fcolorbox{white}{white}{\includegraphics[height=.2\textwidth]{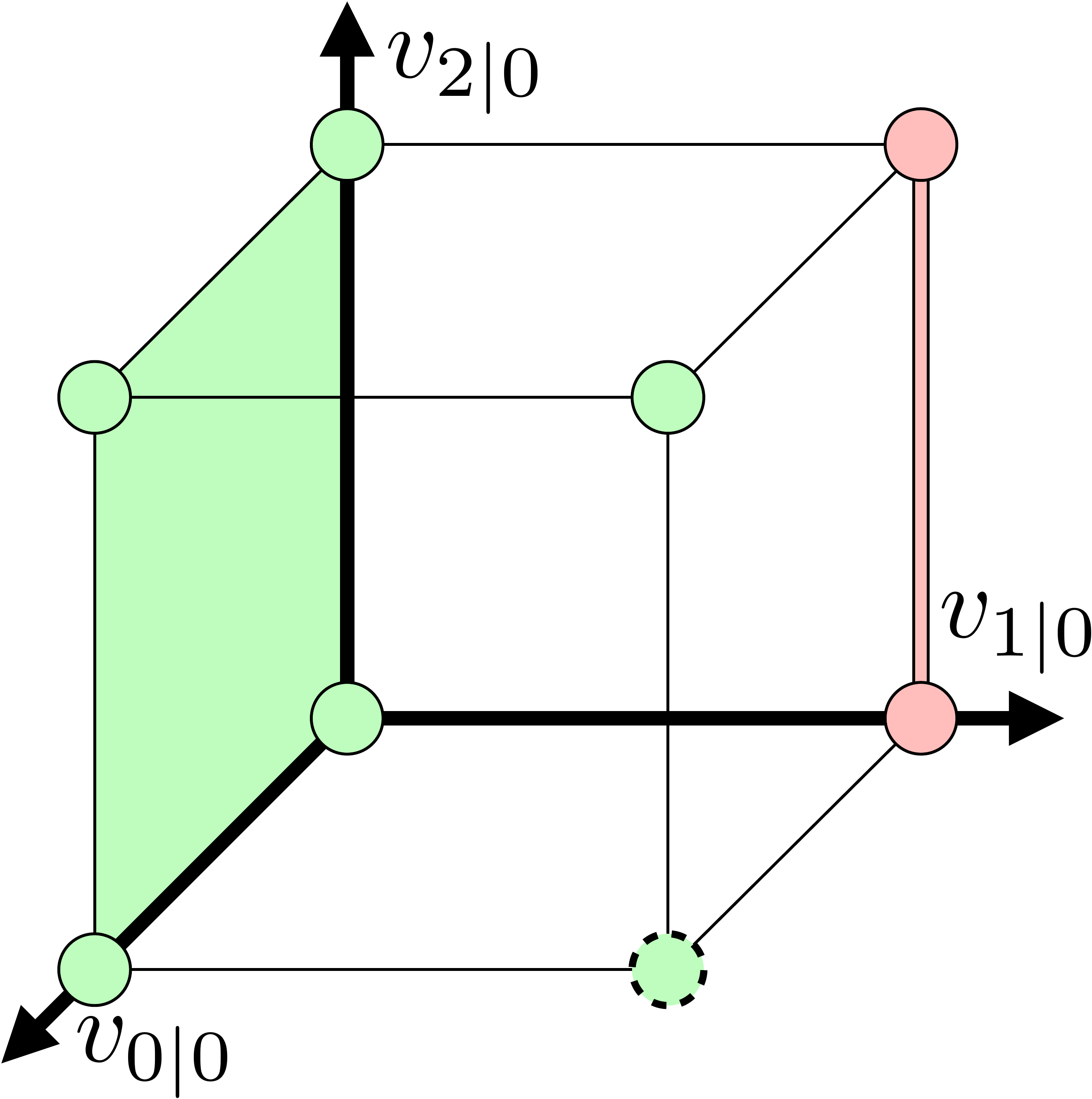}}
\\
\hline
\multicolumn{1}{|c|}{\begin{sideways} Time $\tau=1$ \end{sideways}}
&
\fcolorbox{white}{white}{\includegraphics[height=.2\textwidth]{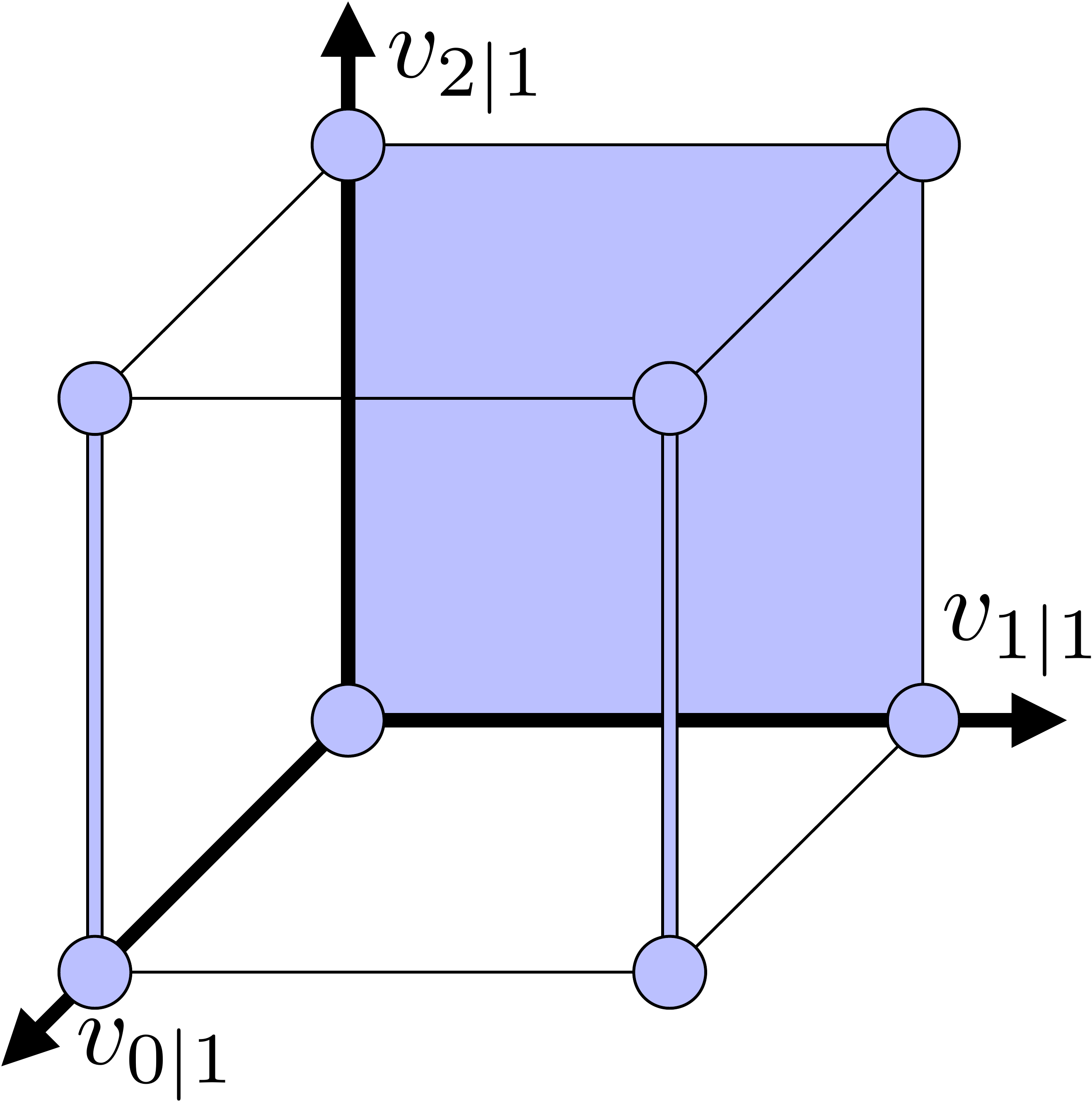}}
&
\fcolorbox{white}{white}{\includegraphics[height=.2\textwidth]{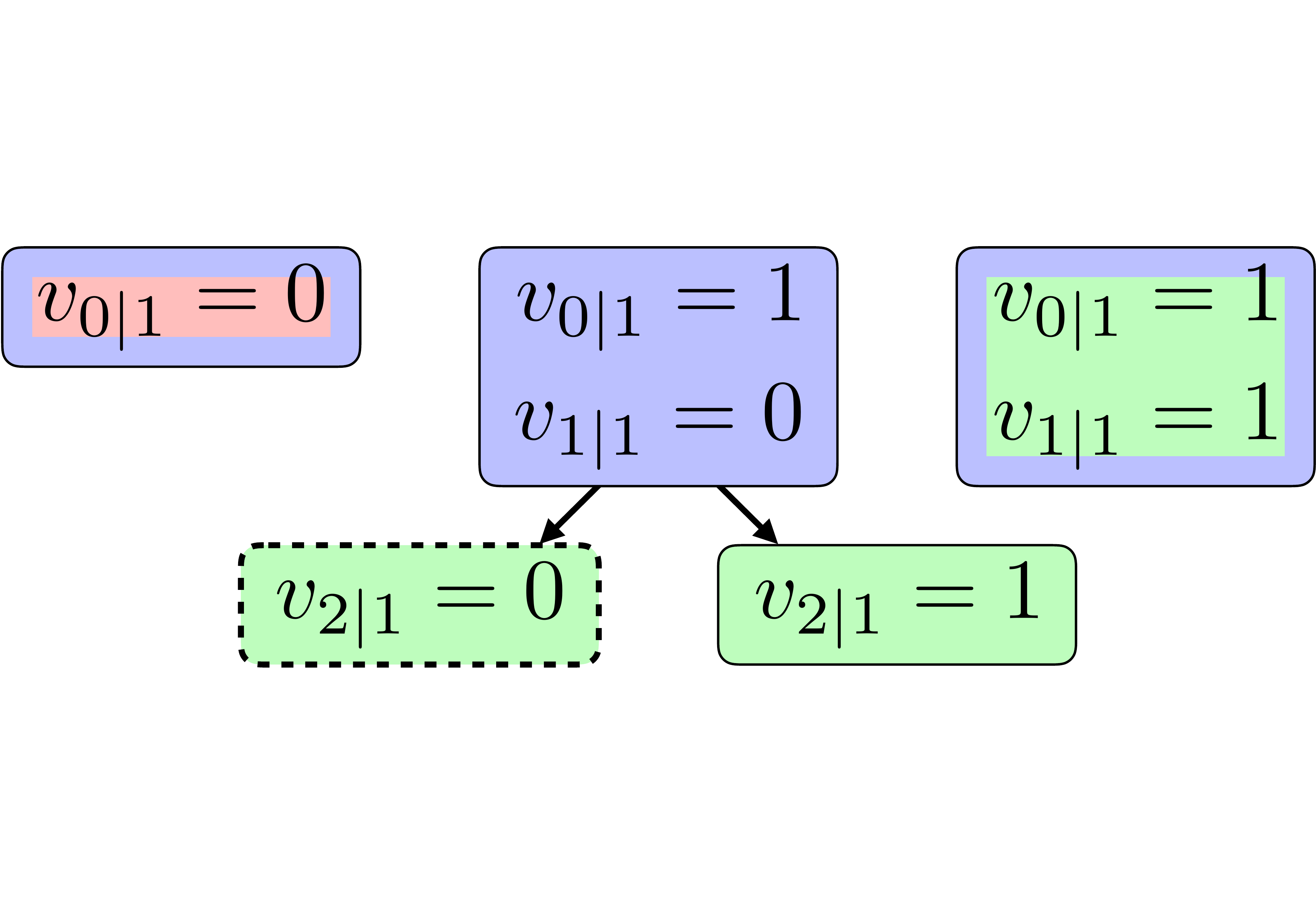}}
&
\fcolorbox{white}{white}{\includegraphics[height=.2\textwidth]{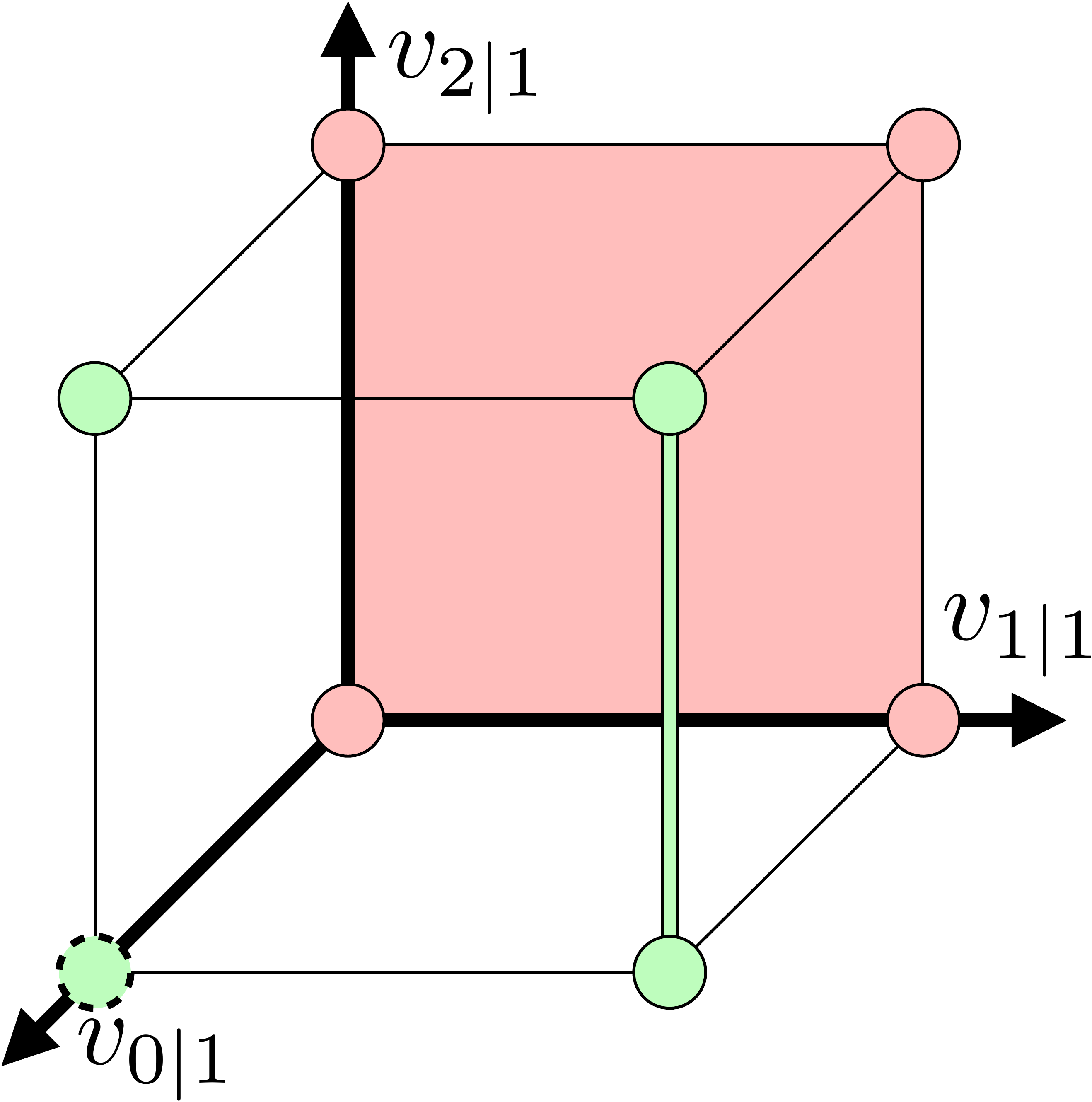}}
\\
\hline
\end{tabular}
\caption{
Illustration of the synthetic Example~\ref{ex:initial_cover_example}.
The first row describes the cold-started solution of the MIQP at time $\tau=0$, reporting the initial cover $\mathscr V_0^0$, the B\&B tree, and the final cover $\mathscr V_0^*$.
In the second row we depict the same elements for the warm-started MIQP at time $\tau=1$.
The optimal binary actions for $\tau=0$ and $\tau=1$ are $v_0 := v\at{0}{0}^* = 1$ and $v_1 := v\at{0}{1}^* = 1$, respectively.
These determine which leaves are kept (green) and which are discarded (red) during the construction of the subsequent initial covers.
Root nodes are colored in light blue, and leaves associated with the optimal solutions have a dashed contour.
The sets in the initial and final covers are colored and contoured in accordance with the B\&B tree.
}
\label{fig:example}
\end{figure*}

We consider a toy problem where the system has a single binary variable $m_u=1$ and the horizon of the controller is $T=3$.
At time $\tau=0$ the B\&B algorithm is initialized with the trivial cover $\mathscr V_0^0 = \{ [(0,0,0),(1,1,1)] \}$ (top-left cell in Figure~\ref{fig:example}).
Assuming the $\varepsilon$-optimal binary assignment to be $(v\at{0}{0}^*, v\at{1}{0}^*, v\at{2}{0}^*) = (1, 1, 0)$, the B\&B tree is shown in the top-center cell.
The root node (light blue) consists in the solution of the subproblem $\mathbf P_0 ([(0,0,0),(1,1,1)])$, whereas the optimal leaf node has a dashed contour and is associated with $\mathbf P_0 ([(1,1,0),(1,1,0)])$.
The final cover for $\mathbf P_0$ is
\begin{multline}
\mathscr V_0^* = \{
[(0,0,0),(1,0,1)],
[(0,1,0),(0,1,1)], \\
[(1,1,0),(1,1,0)],
[(1,1,1),(1,1,1)]
\}
\end{multline}
and is depicted in the top-right cell.

Among all the leaves at time $\tau=0$, the only one that does not verify condition~\eqref{eq:drop_interval_from_cover}, for $v_0 := v\at{0}{0}^* = 1$, is colored in red and represents problem $\mathbf P_0 ([(0,1,0),(0,1,1)])$.
This interval is hence dropped in the construction of the initial cover $\mathscr V_1^0$, while all the other leaves (green) are shifted in time and added to $\mathscr V_1^0$.
(Note that the sets in the final cover are colored and contoured to match the B\&B tree.)
After the time shift~\eqref{eq:shifted_interval} of the bounds, we get the initial cover for  $\mathbf P_1$:
\begin{multline}
\mathscr V_1^0 = \{
[(0,0,0),(0,1,1)],
[(1,0,0),(1,0,1)], \\
[(1,1,0),(1,1,1)]
\},
\end{multline}
which is depicted in the bottom-left cell in Figure~\ref{fig:example}.\footnote{
The shifting process can also be visualized by looking at the covers $\mathscr V_0^*$ and $\mathscr V_1^0$.
First, intersect the intervals in $\mathscr V_0^*$ with the plane $v\at{0}{0} = 1$ and project the resulting sets onto the plane $v\at{1}{0}$, $v\at{2}{0}$.
Then, rename the residual coordinates $v\at{1}{0}$ and $v\at{2}{0}$ as $v\at{0}{1}$ and $v\at{1}{1}$, respectively.
The latter sets are now the projection of $\mathscr V_1^0$ onto the plane $v\at{0}{1}$, $v\at{1}{1}$: the cover $\mathscr V_1^0$ is recovered by extruding them in the $v\at{2}{1}$ direction between $0$ and $1$.
}

The B\&B tree at time $\tau=1$ (bottom-center cell) has three root nodes, one per set in $\mathscr V_1^0$.
Its $\varepsilon$-optimal solution (dashed leaf) is $(v\at{0}{1}^*, v\at{1}{1}^*, v\at{2}{1}^*) = (1, 0, 0)$, and leads to the final cover $\mathscr V_1^*$ depicted in the bottom-right cell.
The same procedure is then applied again to select the leaves for the construction of $\mathscr V_2^0$: green leaves are kept, the red leaf is dropped.

%% file: sections/lower_bounds.tex
The second step in the construction of our warm start is to equip each set in the cover $\mathscr V_1^0$ with a lower bound for the associated minimization problem.
The strategy we adopt is to construct a dual-feasible solution for each of these QPs.

From the solution of $\mathbf P_0$ via B\&B we retrieve the terminal cover $\mathscr V_0^*$ and, for each interval $\mathcal V_0$ in it, we have a feasible solution for the dual subproblem $\mathbf D_0(\mathcal V_0)$.
This can be optimal or just feasible, or even a certificate of infeasibility in case we proved that $\theta_0 (\mathcal V_0) = \infty$.
By means of the construction presented in Section~\ref{sec:initial_cover}, a set $\mathcal V_0$ (if not discarded) is associated with an element $\mathcal V_1$ of the initial cover $\mathscr V_1^0$.
The following lemma shows how a solution of $\mathbf D_0(\mathcal V_0)$ can be shifted in time to comply with the constraints of $\mathbf D_1 (\mathcal V_1)$.

\begin{lemma}
\label{lemma:propagation_dual_feasible_solution}
Let
$\{\lambda\at{t}{0}, \rho\at{t}{0}\}_{t=0}^{T}$ and $\{\mu\at{t}{0}, \ubar \nu\at{t}{0}, \bar \nu\at{t}{0}, \sigma\at{t}{0}\}_{t=0}^{T-1}$ 
be feasible multipliers for $\mathbf D_0(\mathcal V_0)$.
The following set of multipliers is feasible for $\mathbf D_1(\mathcal V_1)$:
\begin{itemize}
\item 
$(\lambda\at{t}{1}, \rho\at{t}{1}) :=  (\lambda\at{t+1}{0}, \rho\at{t+1}{0})$ for $t=0, \ldots, T-1$,
\item
$(\lambda\at{T}{1}, \rho\at{T}{1}) := 0$,
\item 
$(\mu\at{t}{1}, \ubar \nu\at{t}{1}, \bar \nu\at{t}{1}, \sigma\at{t}{1})
:=
(\mu\at{t+1}{0}, \ubar \nu\at{t+1}{0}, \bar \nu\at{t+1}{0}, \sigma\at{t+1}{0})
$ for $t=0, \ldots, T-2$,
\item
$(\mu\at{T-1}{1}, \ubar \nu\at{T-1}{1}, \bar \nu\at{T-1}{1}, \sigma\at{T-1}{1}):= 0$.
\end{itemize}
\end{lemma}

\begin{proof}
We substitute the candidate solution in the constraints of $\mathbf D_1 (\mathcal V_1)$.
From~\eqref{eq:dual_x_t} we obtain
$Q' \rho\at{t+1}{0} + \lambda\at{t+1}{0} - A' \lambda\at{t+2}{0} + F' \mu\at{t+1}{0} = 0$
for $t = 0, \ldots, T-2$, and 
$Q' \rho\at{T}{0} + \lambda\at{T}{0} = 0$.
These constraints are verified by feasibility of the multipliers at time $\tau=0$.
The constraint~\eqref{eq:dual_x_T} holds trivially, as well as condition~\eqref{eq:dual_u_t} for $t=T-1$.
For $t = 0, \ldots, T-2$, the constraint~\eqref{eq:dual_u_t} becomes
$R' \sigma\at{t+1}{0} - B' \lambda\at{t+2}{0} + G' \mu\at{t+1}{0} + V' (\bar \nu\at{t+1}{0} - \ubar \nu\at{t+1}{0}) = 0$,
which holds again by feasibility of the multipliers at time $\tau=0$.
Nonnegativity of $\mu\at{t}{1}$, $\ubar \nu\at{t}{1}$, and $\bar \nu\at{t}{1}$ is ensured by construction.
\end{proof}

Lemma~\ref{lemma:propagation_dual_feasible_solution} has several implications.
Given a set in $\mathscr V_0^*$ and a dual-feasible solution for the associated QP, we can now equip with feasible multipliers, and hence a lower bound, the related set in $\mathscr V_1^0$.
Since we just assumed feasibility of the multipliers at time step $\tau=0$, Lemma~\ref{lemma:propagation_dual_feasible_solution} applies even if, as it is frequently the case, $\mathbf D_0 (\mathcal V_0)$ is not solved to optimality.
Similarly, if the bound we generate for $\mathbf D_1 (\mathcal V_1)$ is tight enough to prevent its solution within the B\&B at time $\tau=1$, the synthesized dual variables can in turn be propagated to bound the optimal value of a subproblem at time $\tau=2$.
On the other hand, if solving $\mathbf D_1 (\mathcal V_1)$ turns out to be necessary, we can still use the multipliers from Lemma~\ref{lemma:propagation_dual_feasible_solution} to warm start this QP solve.
Clearly, as we shift a dual solution across time steps, the tightness of the implied bound will gradually decay, and a few iterations of the QP solver will eventually be required.
However, this is inevitable: the problem on which we are inferring a bound is increasingly different from the one we actually solved.
Finally, note that Lemma~\ref{lemma:propagation_dual_feasible_solution} holds despite any potential model error $e_0$: the current state $x_1=A x_0 + B u_0 + e_0$ appears in the dual problem $\mathbf D_1 (\mathcal V_1)$ as a cost coefficient and, as such, it does not affect dual feasibility.

The following theorem concerns the tightness of the lower bounds we construct via Lemma~\ref{lemma:propagation_dual_feasible_solution}.

\begin{theorem}
\label{th:propagation_bounds}
Let $\{\lambda\at{t}{0}, \rho\at{t}{0}\}_{t=0}^{T}$ and $\{\mu\at{t}{0}, \ubar \nu\at{t}{0}, \bar \nu\at{t}{0}, \sigma\at{t}{0}\}_{t=0}^{T-1}$ be feasible multipliers for $\mathbf D_0 (\mathcal V_0)$ with cost $\ubar \theta_0 (\mathcal V_0)$.
Define
\begin{subequations}
\label{eq:pi}
\begin{align}
\label{eq:pi_1}
\pi_1 := \ & - | Q x_0 |^2 - | R u_0 |^2,
\\
\label{eq:pi_2}
\pi_2 := \ & | \rho\at{0}{0} / 2- Q x_0 |^2 + | \sigma\at{0}{0} /2 - R u_0 |^2,
\\
\label{eq:pi_3} 
\pi_3 := \ & \nonumber (h - F x_0 - G u_0)' \mu\at{0}{0} + (v_0 - \ubar v\at{0}{0})' \ubar \nu\at{0}{0}
\\
&
+ (\bar v\at{0}{0} - v_0)' \bar \nu\at{0}{0},
\\
\label{eq:pi_4}
\pi_4 := \ & - e_0' \lambda\at{1}{0}.
\end{align}
\end{subequations}
The following is a lower bound on $\theta_1 (\mathcal V_1)$:
\begin{align}
\label{eq:lower_bound_nominal}
\ubar \theta_1 (\mathcal V_1) := \ubar \theta_0 (\mathcal V_0) + \sum_{i=1}^4 \pi_i.
\end{align}
\end{theorem}

\begin{proof}
See Appendix~\ref{sec:proof_theorem}.
\end{proof}

Despite the many terms, Theorem~\ref{th:propagation_bounds} is very informative, and an inspection of the expressions in~\eqref{eq:pi} reveals the following.
We recall that, since we are working with lower bounds, we would like these terms to be positive.
\begin{itemize}
\item[$\pi_1$:]
This term represents the MIQP stage cost for $\tau=0$.
It is nonpositive, but this was expected:
standard MPC arguments show that the value function $\theta_\tau$ can actually decrease at this rate (in the absence of disturbances and as the horizon $T$ tends to infinity).
\item[$\pi_2$:]
Recalling the stationarity conditions~\eqref{eq:stationarity_zwt} from Appendix~\ref{sec:dual_qp}, we notice that this nonnegative term vanishes in case the multipliers $\rho\at{0}{0}$, $\sigma\at{0}{0}$ are optimal for $\mathbf D_0 (\mathcal V_0)$, and the control action $u_0$ (injected in the system at time $\tau=0$) is optimal for the subproblem $\mathbf P_0 (\mathcal V_0)$.
\item[$\pi_3$:]
Because of the feasibility of $u_0$, the condition $\ubar v\at{0}{0} \leq v_0 \leq \bar v\at{0}{0}$ imposed in the construction of $\mathscr V_1^0$, and the nonnegativity of $\mu\at{0}{0}$, $\ubar \nu\at{0}{0}$, $\bar \nu\at{0}{0}$, this term is nonnegative.
If these multipliers are optimal for $\mathbf D_0 (\mathcal V_0)$ and $u_0$ is optimal for $\mathbf P_0 (\mathcal V_0)$, this term vanishes by complementary slackness.
\item[$\pi_4$:]
This term is linear in the model error $e_0$.
It is null in case of a perfect model, while it can have either sign in case of discrepancies.
\end{itemize}
Notably, for a perfect model $e_0=0$, the difference $\ubar \theta_1 (\mathcal V_1) - \ubar \theta_0(\mathcal V_0)$ is bounded below by $\pi_1$, which does not depend on the particular pair $\mathbf P_0(\mathcal V_0)$, $\mathbf P_1(\mathcal V_1)$ of subproblems we are considering.

Together with a better insight into the tightness of the bounds we propagate, Theorem~\ref{th:propagation_bounds} also gives us a sufficient condition for the infeasibility of $\mathbf P_1 (\mathcal V_1)$.
The next corollary shows how a certificate of infeasibility for $\mathbf P_0 (\mathcal V_0)$ can be transformed into a certificate for $\mathbf P_1 (\mathcal V_1)$.

\begin{corollary}
\label{cor:propagation_certificate_infeasibility}
Let 
$\{\lambda\at{t}{0}, \rho\at{t}{0}\}_{t=0}^{T}$ and $\{\mu\at{t}{0}, \ubar \nu\at{t}{0}, \bar \nu\at{t}{0}, \sigma\at{t}{0}\}_{t=0}^{T-1}$ 
be a certificate of infeasibility for $\mathbf P_0 (\mathcal V_0)$ with dual objective $\ubar \theta_0 (\mathcal V_0)$.
Then, the set of dual variables defined in Lemma~\ref{lemma:propagation_dual_feasible_solution} is a certificate of infeasibility for $\mathbf P_1 (\mathcal V_1)$ as long as $e_0$ lies in the open halfspace
\begin{align}
\label{eq:E_0}
\lambda\at{1}{0}' e_0 < \ubar \theta_0 (\mathcal V_0) + \pi_3.
\end{align}
Moreover, this inequality is always verified if $e_0=0$.
\end{corollary}

\begin{proof}
We check the definition of a certificate of infeasibility from Section~\ref{sec:duality_subproblem}.
In Lemma~\ref{lemma:propagation_dual_feasible_solution} we  have shown dual feasibility of these multipliers, and, by construction, we have  $\rho\at{t}{1}= 0$ and $\sigma\at{t}{1} =0$ for all $t$.
We are then left to verify positivity of their dual cost $\ubar \theta_1 (\mathcal V_1)$.
Using Theorem~\ref{th:propagation_bounds}, we have $\pi_1 +  \pi_2 = 0$ and $\ubar \theta_1 (\mathcal V_1) = \ubar \theta_0 (\mathcal V_0) + \pi_3 + \pi_4$, which leads to~\eqref{eq:E_0}.
Finally, since $\pi_3 \geq 0$ and $\ubar \theta_0 (\mathcal V_0)> 0$, $e_0=0$ always satisfies this inequality.
\end{proof}

Corollary~\ref{cor:propagation_certificate_infeasibility} completes the tools we need to equip with lower bounds the initial cover $\mathscr V_1^0$.
For any set $\mathcal V_0$ in $\mathscr V_0^*$ that corresponds to an infeasible QP, we can now associate a halfspace in the error space inside which the descendant problem $\mathbf P_1 (\mathcal V_1)$ will also be infeasible.
Moreover, since the set defined by~\eqref{eq:E_0} contains the origin, in case of an exact MLD model, infeasibility of the descendant subproblem is guaranteed.
As for Lemma~\ref{lemma:propagation_dual_feasible_solution}, this process can be iterated and the same certificate propagated across multiple time steps.

Except for the computation of $\pi_4$, which only amounts to $| \mathscr V_1^0 |$ scalar products in $\mathbb R^{n_x}$, all the steps in this section can be performed before the measurement of the next state $x_1$, leading to a negligible time delay in the solution of $\mathbf P_1$.

%% file: sections/upper_bound.tex
The last element we need to warm start the solution of the MIQP $\mathbf P_1$ is an upper bound $\bar \theta_1^0$ on its optimal value.
The natural way to address this problem is to shift the $\varepsilon$-optimal solution of $\mathbf P_0$ and synthesize a feasible solution for $\mathbf P_1$.

The issue of \emph{persistent feasibility} has been widely studied in hybrid MPC (see~\cite[Section~3.5]{lazar2006model} or~\cite[Sections~12.3.1 and~17.8.1]{borrelli2017predictive}).
The standard approach consists in designing the MPC problem so that the terminal state $x\at{T}{\tau}$ lies in a control-invariant set $\mathcal X$ which contains the origin.
More specifically, for all $x$ in $\mathcal X$ there must exist a control action $u \in \mathbb R^{n_u} \times \{0,1\}^{m_u}$ such that  $(x,u) \in \mathcal D$ and $Ax+Bu \in \mathcal X$.
When this is the case, the existence of an input $u\at{T-1}{1}$, such that the control sequence $\{ u\at{1}{0}^*, \ldots, u\at{T-1}{0}^*, u\at{T-1}{1} \}$ is feasible for $\mathbf P_1$, is guaranteed.
The computation of the upper bound $\bar \theta_1^0$ then amounts, in the worst case, to the solution of $2^{m_u}$ QPs.

There are two standard ways to fulfill the requirement $x\at{T}{\tau} \in \mathcal X$, both with well-known pros and cons~\cite{mayne2000constrained}.
The first is to make the MPC horizon $T$ long enough so that the invariance condition is spontaneously verified.
The second is to enforce it explicitly as a terminal constraint in our MIQP.
While the first approach complies with the problem formulation from~\eqref{eq:miqp}, as already mentioned,  the implementation of the second requires a more versatile problem statement which we will consider in Section~\ref{sec:terminal}.

In contrast to Section~\ref{sec:lower_bounds}, here we can only generate upper bounds for a perfect model $e_0 = 0$.
A potential workaround would be to consider a robust version of problem~\eqref{eq:miqp}, where persistent feasibility is guaranteed despite disturbances of bounded magnitude.
This, however, would lead to substantially harder optimization problems (see, e.g.,~\cite{kerrigan2002optimal} or~\cite[Chapter~5]{lazar2006model}).

%% file: sections/asymptotic_analysis.tex
We proceed in the analysis of the warm-start algorithm studying its asymptotic behavior as the horizon $T$ grows.
In doing so, we assume the MLD model to be perfect ($e_0 = 0$).

In order to link the lower bounds from Theorem~\ref{th:propagation_bounds} with the decrease rate of the cost to go $\theta_\tau$, we take advantage of the following observation.

\begin{lemma}
\label{lemma:lower_bound_theta1}
Consider a perfect MLD model~\eqref{eq:mld}.
For any control action $u_0 \in \mathbb R^{n_u} \times \{0,1\}^{m_u}$ applied to the system at time $\tau=0$ such that $(x_0, u_0) \in \mathcal D$, we have $\theta_1 \geq \theta_0 - | Q x_0 |^2 - | R u_0 |^2$.
\end{lemma}
\begin{proof}
Let $\ubar \theta_1 \in \mathbb R_{\geq 0} \cup \{\infty\}$ be the optimal value of the problem we obtain by shortening the horizon of $\mathbf P_1$ by one time step.
Clearly, $\ubar \theta_1 \leq \theta_1$.
On the other hand, we must also have $\theta_0 \leq \ubar \theta_1 + | Q x_0 |^2 + | R u_0 |^2$.
In fact, if this was not true, prepending $u_0$ to the optimal controls from the shortened $\mathbf P_1$ we would get a solution for $\mathbf P_0$ with cost lower than $\theta_0$, a contradiction.
The lemma follows by chaining these two inequalities.
\end{proof}

The following theorems can be seen as ``sanity checks'' for the asymptotic behavior of the warm-start algorithm as $T$ tends to infinity.
More specifically, let $\mathcal V_0^* \in \mathscr V_0^*$ be the set which contains the $\varepsilon$-optimal binary assignment found via B\&B at time $\tau=0$, and denote with $\mathcal V_1^*$ its descendant through the procedure presented in Section~\ref{sec:initial_cover}.
We show that $\mathcal V_1^*$ must contain a binary assignment which is $\varepsilon$-optimal for $\mathbf P_1$.
Moreover, $\varepsilon$-optimality of this assignment is directly proved by the initial cover $\mathscr V_1^0$ from Section~\ref{sec:initial_cover}, equipped with the lower bounds from Theorem~\ref{th:propagation_bounds}.
This formalizes the intuition that, as the horizon grows and the MPC policy tends to be stationary, the warm-started B\&B should only reoptimize the final stage of the trajectory.

\begin{theorem}
\label{th:recursive_cover}
Consider a perfect MLD model~\eqref{eq:mld}, and let the horizon $T$ go to infinity.
The set $\mathcal V_1^*$ contains a binary-feasible assignment for $\mathbf P_1$ with cost $\theta_1^* \leq \theta_1  + \varepsilon$.
\end{theorem}
\begin{proof}
As $T$ goes to infinity, the terminal state $x\at{T}{0}$ of any feasible solution for $\mathbf P_0$ must lie in a control-invariant set $\mathcal X$ within which cost is not accumulated.
More precisely, $ \mathcal X \subseteq \mathrm{ker} (Q)$ and for all $x \in \mathcal X$ there must exist a $u \in \mathbb R^{n_u} \times \{0,1\}^{m_u}$ such that $R u =0$, $(x, u) \in \mathcal D$, and $Ax + Bu \in \mathcal X$.
Thus, the $\varepsilon$-optimal solution of $\mathbf P_0$ with cost $\theta_0^* \leq \theta_0 + \varepsilon$ can be shifted in time to synthesize a feasible solution for $\mathbf P_1$ with cost $\theta_1^* := \theta_0^* - | Q x_0 |^2 - | R u_0 |^2 \leq \theta_0 + \varepsilon - | Q x_0 |^2 - | R u_0 |^2$.
The binaries of the synthesized solution belong to $\mathcal V_1^*$ and, using Lemma~\ref{lemma:lower_bound_theta1}, we get $\theta_1^* \leq \theta_1 + \varepsilon$.
\end{proof}

\begin{theorem}
\label{th:recursive_bound}
Let the assumptions of Theorem~\ref{th:recursive_cover} hold, and let $\theta_1^*$ be defined as in its proof.
The bounds from Theorem~\ref{th:propagation_bounds} verify the condition $\ubar \theta_1 (\mathcal V_1) \geq \theta_1^*  - \varepsilon$ for all $\mathcal V_1$ in $\mathscr V_1^0$.
\end{theorem}
\begin{proof}
Consider a generic set $\mathcal V_1 \in \mathscr V_1^0$ and its ancestor $\mathcal V_0 \in \mathscr V_0^*$.
Since $\pi_2$ and $\pi_3$ from Theorem~\ref{th:propagation_bounds} are nonnegative and $\pi_4 = 0$ by assumption, we have $\ubar \theta_1 (\mathcal V_1) \geq \ubar \theta_0 (\mathcal V_0) + \pi_1$.
By convergence of the B\&B at time $\tau = 0$, we have $ \ubar \theta_0 (\mathcal V_0) \geq \theta_0^* - \varepsilon$ for all $\mathcal V_0$ in $\mathscr V_0^*$ (see condition~\eqref{eq:convergence_bb}).
These imply $\theta_1^* := \theta_0^* + \pi_1 \leq \ubar \theta_0 (\mathcal V_0) + \varepsilon + \pi_1 \leq \ubar \theta_1 (\mathcal V_1 ) + \varepsilon$ for all $\mathcal V_1$ in $\mathscr V_1^0$.
\end{proof}

%% file: sections/lower_bounds_terminal.tex
The main limitation of problem statement~\eqref{eq:miqp} is its incompatibility with terminal penalties and constraints.
In this section, we extend the warm-start algorithm to cost functions and constraints which vary with the relative time $t$.
First, we consider time-dependent weight matrices $Q_t$ and $R_t$ in the objective~\eqref{eq:miqp_objective}.
Once again, we make no assumptions on the rank of these matrices.
Then we replace the constraint set $\mathcal D$ in~\eqref{eq:miqp_constraints} with the time-dependent polyhedron $\mathcal D_t := \{(x,u) \mid F_t x + G_t u \leq h_t \}$, which is assumed to contain the origin and to be contained in $\mathcal D$.
In this wider framework, terminal penalties can be enforced by modulating the value of $Q_T$ and terminal constraints via a suitable definition of $\mathcal D_{T-1}$.
Note that, a polyhedral constraint on $x\at{T}{\tau}$ maps to a polyhedral constraint on $(x\at{T-1}{\tau}, u\at{T-1}{\tau})$ via the dynamics~\eqref{eq:miqp_dynamics}.

Clearly, in these more general settings, the asymptotic analysis from Section~\ref{sec:asymptotic_analysis}, which was based on the limiting stationarity of the MPC policy, might not hold.
In addition, in case of wild variations of the problem data $Q_t$, $R_t$, $\mathcal D_t$  with the relative time $t$, we expect a warm start generated by shifting the previous solution to be fairly ineffective.
However, as we show in this section, our algorithm deals with these issues very transparently, propagating dual bounds that are parametric in the variations of these problem data.

\subsection{Stage Cost Varying with the Relative Time}
\label{sec:assumption_cost}
\input{sections/assumption_cost}

\subsection{Constraint Set Varying with the Relative Time}
\label{sec:assumption_constraints}
\input{sections/assumption_constraints}

%% file: sections/assumption_cost.tex
We start discussing the implications of time-dependent weight matrices $Q_t$ and $R_t$.
We do so under the following assumption which, in words, requires the weight matrices for time $t+1$ to penalize only the state and input entries that are also penalized at stage $t$.

\begin{assumption}
\label{ass:row_space_cost}
The row space of $Q_t$ contains the row space of $Q_{t+1}$ for $t= 0, \ldots, T-1$, and the row space of $R_t$ contains the row space of $R_{t+1}$ for $t= 0, \ldots, T-2$.
\end{assumption}

Out of the three components of our warm-start procedure, only the propagation of lower bounds (discussed in Section~\ref{sec:lower_bounds}) is affected by the time dependency of $Q_t$ and $R_t$.
Here, we extend this component starting from Lemma~\ref{lemma:propagation_dual_feasible_solution}, and then considering Theorem~\ref{th:propagation_bounds} and Corollary~\ref{cor:propagation_certificate_infeasibility}.

\subsubsection{Modifications to Lemma~\ref{lemma:propagation_dual_feasible_solution}}
\label{sec:on_lemma_1_cost}

Following the steps from Appendix~\ref{sec:dual_qp}, the dual problem~\eqref{eq:dual} can be easily adjusted to comply with time-varying weights.
Constraints~\eqref{eq:dual_x_t} and~\eqref{eq:dual_u_t} require the substitution of $Q$ and $R$ with $Q_t$ and $R_t$, respectively.
Constraint~\eqref{eq:dual_x_T} now depends on $Q_T$ instead of $Q$.
This modification breaks the shifting procedure presented in Lemma~\ref{lemma:propagation_dual_feasible_solution}.
To restore it, we redefine $\rho\at{t}{1}$ for $t= 0, \ldots, T-1$ and $\sigma\at{t}{1}$ for $t= 0, \ldots, T-2$, explicitly enforcing the conditions $Q_t' \rho\at{t}{1} = Q_{t+1}' \rho\at{t+1}{0}$  and $R_t' \sigma\at{t}{1} = R_{t+1}' \sigma\at{t+1}{0}$.
Furthermore, among all the solutions of these linear systems, we select the ones that maximize the lower bound $\ubar \theta_1 (\mathcal V_1)$ or, equivalently,  minimize $|\rho\at{t}{1}|^2$ and $|\sigma\at{t}{1}|^2$.
This choice leads to two quadratic optimization problems, which, under Assumption~\ref{ass:row_space_cost}, are always feasible and admit closed-form solution:
\begin{subequations}
\label{eq:rho_sigma}
\begin{align}
& \rho\at{t}{1} := (Q_t')^+ Q_{t+1}' \rho\at{t+1}{0}, && t= 0, \ldots, T-1,\\
& \sigma\at{t}{1} := (R_t')^+ R_{t+1}' \sigma\at{t+1}{0}, && t= 0, \ldots, T-2.
\end{align}
\end{subequations}

\subsubsection{Modifications to Theorem~\ref{th:propagation_bounds} and Corollary~\ref{cor:propagation_certificate_infeasibility}}

Theorem~\ref{th:propagation_bounds} can be adapted to the time dependency of $Q_t$ and $R_t$ by retracing the steps from Appendix~\ref{sec:proof_theorem}.
The definitions in~\eqref{eq:pi} are still valid, provided that we substitute the matrices $Q$ and $R$ with $Q_0$ and $R_0$.
The lower bound $\ubar \theta_1 (\mathcal V_1)$ from~\eqref{eq:lower_bound_nominal} requires the addition of two terms:
\begin{subequations}
\label{eq:pi56}
\begin{align}
\pi_5 := & \frac{1}{4} \sum_{t=0}^{T-1} ( |\rho\at{t+1}{0} |^2 - |\rho\at{t}{1} |^2), \\
\pi_6 := & \frac{1}{4} \sum_{t=0}^{T-2} ( |\sigma\at{t+1}{0} |^2 - |\sigma\at{t}{1} |^2),
\end{align}
\end{subequations}
which do not cancel out anymore.
In the following propositions we analyze the sign of $\pi_5$ and $\pi_6$.

\begin{proposition}
\label{prop:nonnegativity_pi5}
A necessary and sufficient condition for $\pi_5$ to be nonnegative for all $\rho\at{1}{0}, \ldots, \rho\at{T}{0}$ is $\| Q_{t+1} Q_t^+ \| \leq 1$ for $t=0, \ldots, T-1$.
\end{proposition}

\begin{proof}
By definition of the operator norm, we have $| \rho\at{t}{1} | \leq \| Q_{t+1} Q_t^+ \| | \rho\at{t+1}{0} |$.
Sufficiency follows from
\begin{align}
\label{eq:lower_bound_pi5}
\pi_5 \geq \frac{1}{4} \sum_{t=0}^{T-1} ( 1 - \| Q_{t+1} Q_t^+ \|^2 ) |\rho\at{t+1}{0}|^2.
\end{align}
For the other direction, note that equality in~\eqref{eq:lower_bound_pi5} can always be attained for some nonzero $\rho\at{1}{0}, \ldots, \rho\at{T}{0}$.
\end{proof}

\begin{proposition}
A necessary and sufficient condition for $\pi_6$ to be nonnegative for all $\sigma\at{1}{0}, \ldots, \sigma\at{T-1}{0}$ is $\| R_{t+1} R_t^+ \| \leq 1$ for $t=0, \ldots, T-2$.
\end{proposition}

\begin{proof}
Analogous to the proof of Proposition~\ref{prop:nonnegativity_pi5}.
\end{proof}

These propositions suggest that, when the magnitude of the weight matrices $Q_t$ and $R_t$ increases with the relative time $t$, the terms $\pi_5$ and $\pi_6$ might be negative.
On the other hand, when the weights decrease with $t$ (or when they are constant as in Theorem~\ref{th:propagation_bounds}), $\pi_5$ and $\pi_6$ tend to tighten the bounds $\ubar \theta_1 (\mathcal V_1)$.
Unfortunately, terminal penalties fall under the first scenario, but this was to be expected: since the final state of $\mathbf P_1 (\mathcal V_1)$ is likely to be smaller in magnitude than the one of $\mathbf P_0 (\mathcal V_0)$, when increasing $Q_T$, we expect the difference $\theta_1 (\mathcal V_1) - \theta_0 (\mathcal V_0)$ to decrease.

Finally, Corollary~\ref{cor:propagation_certificate_infeasibility} remains unchanged as well as inequality~\eqref{eq:E_0}.
In fact, if the multipliers we are given for $\tau=0$ certify infeasibility of $\mathbf P_0 (\mathcal V_0)$, then $\rho\at{t}{0} = 0$ and $\sigma\at{t}{0} = 0$ for all $t$.
 Using~\eqref{eq:rho_sigma}, we get $\rho\at{t}{1} = 0$ and $\sigma\at{t}{1} = 0$, and the additional terms $\pi_5$ and $\pi_6$ from~\eqref{eq:pi56} vanish.

%% file: sections/assumption_constraints.tex
We discuss time-dependent constraints under the following assumption, which holds, for example, in the common case of bounded constraint sets $\mathcal D_t$.

\begin{assumption}
\label{ass:row_space_constraints}
The conic hull of the rows of $[F_t \ G_t ]$ contains the conic hull of the rows of $[F_{t+1} \ G_{t+1}]$ for $t= 0, \ldots, T-2$.
\end{assumption}

Out of the three warm start components, the time dependency of $\mathcal D_t$ affects mainly the propagation of lower bounds.
The construction of the initial cover is unchanged and, for the upper-bound propagation, in order to preserve persistent feasibility, we only require $\mathcal D_{t+1} \subseteq \mathcal D_t$ for $t=0, \ldots, T-2$.
Once again we discuss the adaptations of Lemma~\ref{lemma:propagation_dual_feasible_solution} and of Theorem~\ref{th:propagation_bounds} and Corollary~\ref{cor:propagation_certificate_infeasibility} separately.

\subsubsection{Modifications to Lemma~\ref{lemma:propagation_dual_feasible_solution}}

In case of time-dependent constraints $\mathcal D_t$, the matrices $F$, $G$, and $h$ in the dual problem~\eqref{eq:dual} must have the subscript $t$.
This modification makes the arguments from Lemma~\ref{lemma:propagation_dual_feasible_solution} untrue.
The ideal fix would be to define $\mu\at{t}{1}$ through the Linear Program (LP)
\begin{subequations}
\label{eq:lp}
\begin{align}
\min \
& h_t' \mu\at{t}{1} \\
\mathrm{s.t.} \ &
\label{eq:lp_1}
F_t' \mu\at{t}{1} = F_{t+1}' \mu\at{t+1}{0}, \\
\label{eq:lp_2}
& G_t' \mu\at{t}{1} = G_{t+1}' \mu\at{t+1}{0}, \\
\label{eq:lp_3}
& \mu\at{t}{1} \geq 0,
\end{align}
\end{subequations}
 for $t=0, \ldots, T-2$.
These LPs would maximize the lower bounds $\ubar \theta_1 (\mathcal V_1)$ and, under Assumption~\ref{ass:row_space_constraints}, they would always be feasible.
However, keeping in mind that our ultimate goal is to bound the optimal value of a QP, this definition of $\mu\at{t}{1}$ is clearly impractical.
Nonetheless, finding a good approximate solution to these LPs turns out to be relatively simple.

Let $\mu\at{t}{1}^* (\mu\at{t+1}{0})$ be the parametric minimizer of problem~\eqref{eq:lp}.
We define $M\at{t}{1}$ as the matrix whose $i$th column is $\mu\at{t}{1}^*  (\epsilon_i)$, with $\epsilon_i$ $i$th element of the standard basis.
Note that $M\at{t}{1}$ can be easily computed offline by solving one LP per entry in $\mu\at{t+1}{0}$ (i.e., per facet of the polyhedron $\mathcal D_{t+1}$).

\begin{proposition}
\label{prop:feasible_sol_lp}
The multiplier $ \mu\at{t}{1} := M\at{t}{1} \mu\at{t+1}{0}$ is feasible for the LP~\eqref{eq:lp}.
\end{proposition}
\begin{proof}
Since $\mu\at{t+1}{0}$ and $\mu\at{t}{1}^* (\epsilon_i)$ are nonnegative, so is $ \mu\at{t}{1}$.
By feasibility of $\mu\at{t}{1}^*$, we have $F_t' M\at{t}{1} = F_{t+1}'$ and $G_t' M\at{t}{1} = G_{t+1}'$, which imply conditions~\eqref{eq:lp_1} and~\eqref{eq:lp_2}.
\end{proof}

Coming back to the primal side, the LP~\eqref{eq:lp} has a clear geometrical meaning.
Its dual reads
\begin{subequations}
\label{eq:lp_primal}
\begin{align}
\max \
& \mu\at{t+1}{0}' (F_{t+1} x\at{t}{1} + G_{t+1} u\at{t}{1}) \\
\mathrm{s.t.} \ &
F_t x\at{t}{1} + G_t u\at{t}{1} \leq h_t,
\end{align}
\end{subequations}
where the optimization variables are the state $x\at{t}{1}$ and the input $u\at{t}{1}$.
For $\mu\at{t+1}{0} = \epsilon_i$, the LP~\eqref{eq:lp_primal} is illustrated in Figure~\ref{fig:containment_problem}, and allows to determine whether the polyhedron $\mathcal D_t$ lies within the halfspace delimited by the $i$th facet of $\mathcal D_{t+1}$.
In words, this optimization finds the point in~$\mathcal D_t$ which violates most the $i$th inequality defining $\mathcal D_{t+1}$.
Containment is certified in case the maximum of this problem, which corresponds to $h_t' \mu\at{t}{1}^* (\epsilon_i)$ by strong duality, is lower or equal to the $i$th entry of $h_{t+1}$.\footnote{
In the context of problem~\eqref{eq:lp_primal}, Assumption~\ref{ass:row_space_constraints} has a simple geometrical interpretation as well.
It ensures that the normal to each facet of $\mathcal D_{t+1}$ is not a ray of the polyhedron $\mathcal D_t$, i.e., it ensures boundedness of~\eqref{eq:lp_primal}.
Moreover, note that feasibility of~\eqref{eq:lp_primal} (and hence boundedness of~\eqref{eq:lp}) is also ensured, since we assumed the polyhedra $\mathcal D_t$ to be nonempty for all $t$.
}

If the polyhedron $\mathcal D_t$ is entirely contained in $\mathcal D_{t+1}$, the above observation applies for all $i$, and we have $h_t' M\at{t}{1} \leq h_{t+1}'$.
This inequality can in turn be used to bound the cost of $ \mu\at{t}{1}$ from Proposition~\ref{prop:feasible_sol_lp}, leading to
\begin{align}
\label{eq:upper_bound_beta}
h_t' \mu\at{t}{1} =
h_t' M\at{t}{1} \mu\at{t+1}{0} \leq
h_{t+1}' \mu\at{t+1}{0}.
\end{align}
We will take advantage of this bound in the revision of Theorem~\ref{th:propagation_bounds}.

\begin{figure}[t]
\centering
\includegraphics[width=.32\textwidth]{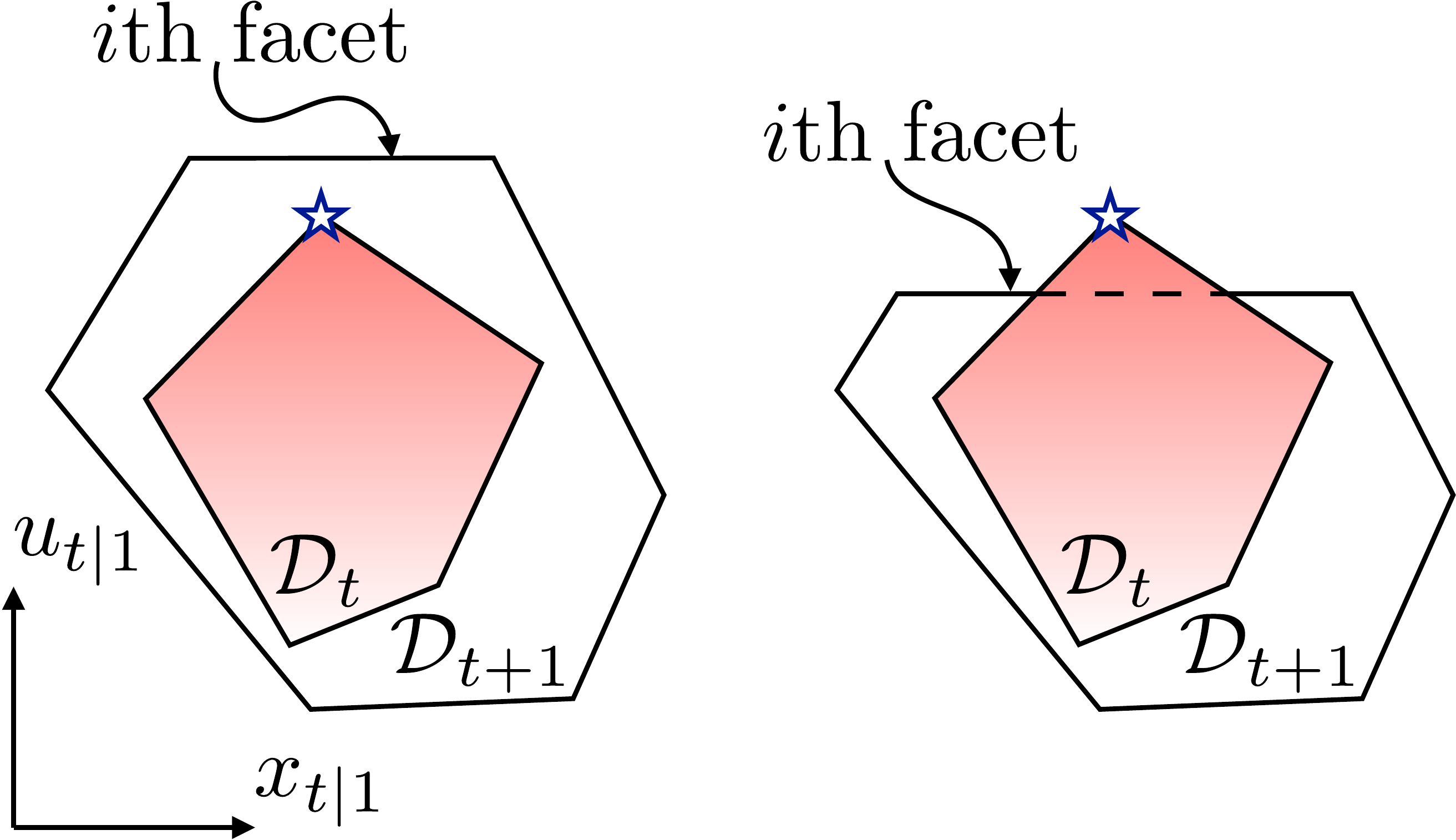}
\caption{
Geometrical interpretation of the LP~\eqref{eq:lp_primal} as a containment problem.
The color gradient in $\mathcal D_t$ symbolizes the objective function.
For $\mu\at{t+1}{0} = \epsilon_i$, problem~\eqref{eq:lp_primal} returns the point  (blue star) in $\mathcal D_t$ which violates most the $i$th constraint (facet) of $\mathcal D_{t+1}$.
Depending on whether the polyhedron $\mathcal D_t$ lies inside the $i$th facet of $\mathcal D_{t+1}$, the optimal value of~\eqref{eq:lp}, and of its dual~\eqref{eq:lp_primal}, is lower (left image) or greater (right image) than the $i$th entry of $h_{t+1}$.
}
\label{fig:containment_problem}
\end{figure}

\subsubsection{Modifications to Theorem~\ref{th:propagation_bounds} and Corollary~\ref{cor:propagation_certificate_infeasibility}}

When the sets $\mathcal D_t$ vary with the relative time $t$, the matrices $F$, $G$, and $h$ in the definition of $\pi_3$ must be substituted with $F_0$, $G_0$, and $h_0$.
The lower bounds $\ubar \theta_1 (\mathcal V_1)$ from Theorem~\ref{th:propagation_bounds} require the additional term
\begin{align}
\pi_7 :=
\sum_{t=0}^{T-2} (h'_{t+1} \mu\at{t+1}{0} - h'_t \mu\at{t}{1}).
\end{align}

The observation made in~\eqref{eq:upper_bound_beta} suggests the following sufficient condition for the nonnegativity of $\pi_7$.

\begin{proposition}
\label{prop:pi_7}
Let $\mu\at{t}{1}$ be defined as in Proposition~\ref{prop:feasible_sol_lp}.
If $\mathcal D_t \subseteq \mathcal D_{t+1}$ for $t= 0, \ldots, T-2$, then $\pi_7 \geq 0$.
Additionally, if $\mathcal D_t = \mathcal D_{t+1}$, we have $\pi_7=0$.
\end{proposition}

\begin{proof}
The nonnegativity condition follows from~\eqref{eq:upper_bound_beta}.
In case $\mathcal D_t = \mathcal D_{t+1}$, the optimal value of~\eqref{eq:lp} for $\mu\at{t+1}{0} = \epsilon_i$ coincides with the $i$th entry of $h_{t+1}$, for all $i$.
Therefore, we have $h'_t M\at{t}{1} = h_{t+1}'$, the relation in~\eqref{eq:upper_bound_beta} holds with the equality, and $\pi_7$ vanishes.
\end{proof}

Even if the condition $\mathcal D_t \subseteq \mathcal D_{t+1}$ is frequently violated in practice (terminal constraints, for example, lead to $\mathcal D_{T-2} \supset \mathcal D_{T-1}$),  Proposition~\ref{prop:pi_7} shows that the definition of $\mu\at{t}{1}$ from Proposition~\ref{prop:feasible_sol_lp} is a natural generalization of the shifting process from Lemma~\ref{lemma:propagation_dual_feasible_solution}.
In fact, when the constraint sets $\mathcal D_t$ are actually constant with the relative time $t$, the two approaches lead to the same lower bounds $\ubar \theta_1 (\mathcal V_1)$.

With this choice of the multipliers $\mu\at{t}{1}$, the statement of Corollary~\ref{cor:propagation_certificate_infeasibility} is still valid, provided that we add $\pi_7$ to the right-hand side of~\eqref{eq:E_0}.
Furthermore, if  $\mathcal D_t \subseteq \mathcal D_{t+1}$ for all $t$, the origin $e_0=0$ is still guaranteed to verify condition~\eqref{eq:E_0}.
On the contrary, if the constraint sets shrink with the relative time $t$, it might be the case that an infeasible subproblem at time $\tau=0$ has a feasible descendant at $\tau=1$, even in the nominal case $e_0=0$.

%% file: sections/numerical_study.tex
We test the proposed warm-start algorithm on a numerical example.
We consider a linearized version of the cart-pole system depicted in Figure~\ref{fig:cart_pole}: the goal is to regulate the cart in the center of the two walls with the pole in the upright position.
To accomplish this task, we can apply a force directly on the cart and exploit contact forces that arise when the tip of the pole collides with the walls.
This regulation problem has been used to benchmark control-through-contact algorithms in~\cite{marcucci2017approximate,deits2019lvis}, and its moderate size allows an in-depth statistical analysis of the performance of our warm-start technique.

\begin{figure}[t]
\centering
{\includegraphics[width = .7\columnwidth]{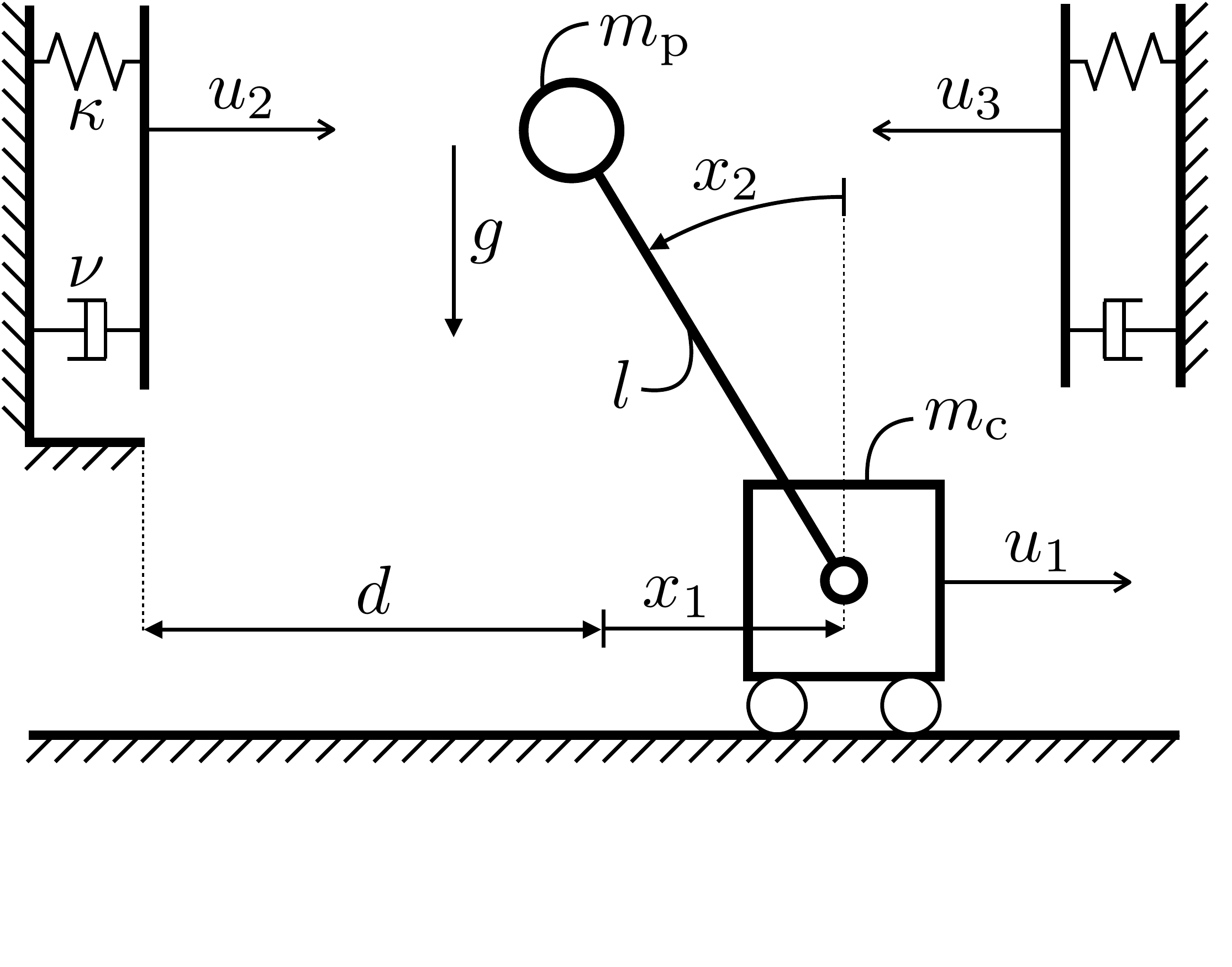}}
\caption{Benchmark problem: regulation of the cart-pole system through a force applied to the cart and exploiting contacts with the soft walls.}
\label{fig:cart_pole}
\end{figure}

\subsection{Mixed Logical Dynamical Model}

We let $x_1$ be the position of the cart, $x_2$ the angle of the pole, and we denote with $x_3$ and $x_4$ their time derivatives.
The force applied to the cart is $u_1$, whereas the contact forces with the left and right walls are $u_2$ and $u_3$, respectively.
The continuous-time equations of motion, linearized around the nominal angle of the pole $x_2=0$, are
\begin{subequations}
\label{eq:cart_pole_dynamics}
\begin{align}
\dot x_1 &= x_3, \\
\dot x_2 &= x_4, \\
\dot x_3 &= \frac{m_{\mathrm p} g}{m_{\mathrm c}} x_2 + \frac{1}{m_{\mathrm c}} u_1, \\
\dot x_4 &= \frac{(m_{\mathrm c} + m_{\mathrm p}) g}{m_{\mathrm c} l} x_2 + \frac{1}{m_{\mathrm c} l} u_1 - \frac{1}{m_{\mathrm p} l} u_2 + \frac{1}{m_{\mathrm p} l} u_3,
\end{align}
\end{subequations}
with $m_{\mathrm c}=m_{\mathrm p}=1$ mass of the cart and the pole, $g=10$ gravity acceleration, and $l=1$ length of the pole.
Dynamics are discretized using the explicit Euler method with time step $h=0.05$.
The force applied to the cart and the system state are subject to the constrains $\ubar u_1 \leq u_1 \leq \bar u_1$ and $\ubar x \leq x \leq \bar x$, where $\bar u_1 = - \ubar u_1 = 1$, $\bar x = - \ubar x = ( d, \pi/10, 1, 1)$ , and $d=0.5$ is half of the distance between the walls (see Figure~\ref{fig:cart_pole}).

Impacts between the pole and the walls are modeled with soft contacts: $\kappa = 100$ is the stiffness and $\nu =10$  is the damping in the contact model.
The position of the tip of the pole with respect to the walls (positive in case of penetration), after linearization, is $\delta_2 := - x_1 + l x_2 - d$ for the left wall, and $\delta_3 := x_1 - lx_2 - d$ for the right wall.
For $i\in\{2,3\}$, contact forces are required to obey the constitutive model
\begin{align}
\label{eq:contact_model}
u_i
=
\begin{cases}
\kappa \delta_i + \nu \dot \delta_i & \text{if} \ \delta_i \geq 0 \ \text{and} \ \kappa \delta_i + \nu \dot \delta_i \geq 0, \\
0 & \text{otherwise}.
\end{cases}
\end{align}
These conditions ensure that contact forces are nonzero only in case of penetration, and are always nonnegative (i.e., the walls never pull on the pole).
To model these piecewise-linear functions, we introduce two binary indicators per contact
\begin{align}
\label{eq:binaries_contacts}
u_{i+2}
:=
\begin{cases}
1 & \text{if} \ \delta_i \geq 0, \\
0 & \text{otherwise},
\end{cases}
\quad
u_{i+4}
:=
\begin{cases}
1 & \text{if} \ \kappa \delta_i + \nu \dot \delta_i \geq 0, \\
0 & \text{otherwise}.
\end{cases}
\end{align}
By means of the state limits, we can derive explicit bounds $\ubar \delta_i, \bar \delta_i$ on the penetrations, as well as on their time derivatives $\ubar {\dot \delta}_i, \bar {\dot \delta}_i$.
These, in turn, are used to bound the contact forces with $\ubar u_i := \kappa \ubar \delta_i + \nu \ubar {\dot \delta}_i$ and $\bar u_i := \kappa \bar \delta_i + \nu \bar {\dot \delta}_i$.
Conditions~\eqref{eq:binaries_contacts} are then enforced through the linear inequalities
\begin{subequations}
\begin{align}
\ubar \delta_i (1 - u_{i+2})
\leq
\delta_i
\leq
\bar \delta_i u_{i+2},
\\
\ubar u_i (1 - u_{i+4})
\leq
\kappa \delta_i + \nu \dot \delta_i
\leq
\bar u_i u_{i+4}.
\end{align}
\end{subequations}
With a similar logic, we can express~\eqref{eq:contact_model} through the conditions: $u_i \geq 0$, $u_i \leq \bar u_i u_{i+2}$, $u_i \leq \bar u_i u_{i+4}$, and
\begin{align}
\nu \bar {\dot \delta}_i (u_{i+2} - 1)
\leq
u_i - \kappa \delta_i - \nu \dot \delta_i
\leq
\ubar u_i (u_{i+4} - 1).
\end{align}

Considering the binary inputs introduced as contact indicators, we have an MLD system with $n_x=4$ states, $n_u=3$ continuous inputs, and $m_u=4$ binary inputs.

\subsection{Model Predictive Controller}

We synthesize an MPC controller featuring both a terminal penalty and a terminal constraint (see Section~\ref{sec:terminal}).
For the stage cost, we let $Q_t = I$ and $R_t =(1, 0, 0, 0, 0, 0, 0)'$ for $t=0, \ldots, T-1$.
Using these weights and setting $u_2 = u_3 = 0$, the terminal penalty $Q_T$ is obtained by solving the Discrete Algebraic Riccati Equation (DARE) for the discretized version of~\eqref{eq:cart_pole_dynamics}.
The terminal set is the maximal positive-invariant set for system~\eqref{eq:cart_pole_dynamics} after discretization, in closed loop with the controller from the DARE and subject to the input and state bounds, and the nonpenetration constraints $\delta_i \leq 0$ for $i=2,3$.\footnote{
This set is known to be a polyhedron~\cite{gilbert1991linear} and, in this case, it has a finite number of facets.
See also~\cite[Definition~10.8]{borrelli2017predictive}.
}
With a time horizon $T=20$, the resulting MIQP has 224 optimization variables (144 continuous and 80 binaries) and 906 linear constraints (84 equalities and 822 inequalities).

\subsection{Branch-and-Bound Implementation}
\label{sec:bb_impl}

The results we present in this section are obtained with a very basic \texttt{python} implementation of B\&B, which follows to the letter the description given in Section~\ref{sec:bb_algorithm}.
This has the advantage of simplifying result interpretation, since it leaves out of the analysis the many heuristics that come into play when using advanced B\&B solvers.
Nevertheless, we underline that MILP reoptimization techniques similar in nature to the one we propose have been successfully integrated, e.g., with the state-of-the-art solver \texttt{SCIP}~\cite{achterberg2009scip} in~\cite{gamrath2015reoptimization}.
The latter work shows how advanced B\&B routines (such as presolving, domain propagation, and strong branching) can be handled when reusing the B\&B frontier from the previous solves.

In our implementation, we adopt a best-first search: among all the sets which verify condition~\eqref{eq:convergence_bb}, we pick the set $\mathcal V^i \in \mathscr V^i$ for which $\ubar \theta (\mathcal V^i)$ is minimum.
We perform the branching step in chronological order: each time this routine is called, we select the relaxed variables $v\at{t}{\tau}$ for which $t$ is lowest and, among these, we split the one with the smallest index.
This frequently-used heuristic, leveraging the control limits, quickly rules out excessively fast mode transitions~\cite{bemporad1999efficient,frick2015embedded}.
Since we include model errors in the upcoming analysis, the recursive-feasibility arguments from Section~\ref{sec:upper_bound} do not apply, and we let $\bar \theta^0 = \infty$ in all the B\&B solves.
The optimality tolerance $\varepsilon$ in~\eqref{eq:convergence_bb} is set to zero.
QPs are solved using the dual simplex method provided by the commercial solver \texttt{Gurobi~9.0.1}, with default options.
Root-node subproblems are warm started as discussed in Sections~\ref{sec:lower_bounds} and~\ref{sec:terminal} while, deeper in the B\&B tree, subproblems are warm started using the parent active set.

\subsection{Statistical Analysis}

We test the warm-start algorithm in a ``push-recovery'' task where, to simulate a push towards the right wall, we set the initial state to $x_0 := (0,0,1,0)$.
Assuming a perfect model, Figure~\ref{fig:simulation} depicts the optimal control sequence, and the related trajectory of the tip of the pole, for a closed-loop simulation of $50$ steps.
The system exploits the (soft) right wall to decelerate and come back to the center of the track, whereas the control requires a significant saturation to accomplish the task.

\begin{figure}[t]
\centering
{\includegraphics[width = 1.\columnwidth]{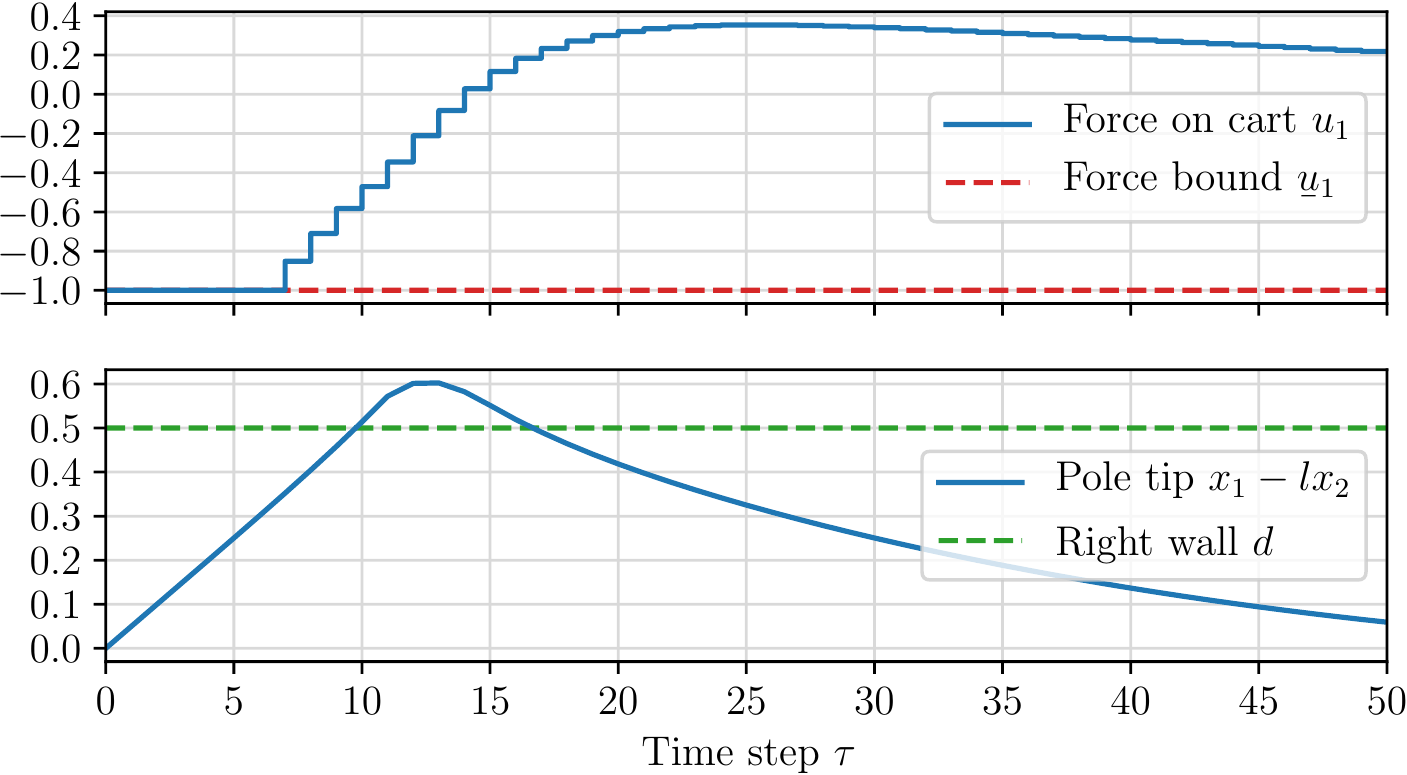}}
\caption{
Optimal closed-loop trajectories for the cart-pole system recovering from a push towards the right wall.
Top: input force applied to the cart.
Bottom: horizontal position of the tip of the pole.
The penetration of the pole in the right wall is allowed by the soft contact model.
}
\label{fig:simulation}
\end{figure}

We study this task in presence of random model errors.
At each time $\tau$ we draw the $i$th entry of the error $e_\tau := x_{\tau+1} - A x_\tau - B u_\tau$ from the normal distribution with zero mean and standard deviation $\sigma_i = c \bar x_i$, with $\bar x_i$ upper bound on the $i$th state.
For $c = 10^{-3}, 3 \cdot 10^{-3}, 10^{-2}$, we simulate 100 closed-loop trajectories (for which model errors do not drive the system to an infeasible state) and we monitor the number of QPs solved within B\&B and the MIQP solution times.

\subsubsection{Number of Branch-and-Bound Subproblems}

We start by comparing the number of QPs solved within B\&B in case of warm and cold start (i.e., when each MIQP is solved from scratch).
Furthermore, to show that the amount of information propagated by the warm starts does not diverge as more and more MIQPs are solved, we analyze the cardinality of the initial covers $\mathscr V_\tau^0$.
For these three quantities, Figure~\ref{fig:nodes} reports the minimum, maximum, 80th and 90th percentile of the values registered in the 100 trials.
Additionally, Figure~\ref{fig:nodes} shows the results obtained in the nominal case, $e_\tau = 0$ for all $\tau$.

For small model errors, the warm-start approach almost always requires an order of magnitude less QPs to solve problem~\eqref{eq:miqp} to global optimality.
Note that, for $c=3 \cdot 10^{-3}$, the curve of the 90th percentile almost coincides with the one of the minima.
For $c= 10^{-2}$ model errors become very significant: we often have mismatches in the cart position greater than $10^{-2}$ which, if multiplied by $\kappa$, lead to variations of the contact forces, with respect to the planned values, greater than the input limit $\bar u_1 = 1$.
Despite that, $80\%$ of the times the proposed technique reduces the number of QP solves by an order of magnitude.
Moreover, even in the worst case, our warm-start algorithm outperforms the cold-start approach.\footnote{
We report that, trying to further increase the error standard deviation by setting $c=3 \cdot 10^{-2}$, the model errors drive the system to an infeasible state 98 times on 100 trials, generating statistics of little value.
}

The asymptotic behavior discussed in Section~\ref{sec:asymptotic_analysis} is also found in Figure~\ref{fig:nodes}.
To solve a problem with $m$ binaries, in fact, the minimum number of B\&B subproblems is $2 m + 1$ (the optimal branch plus the necessary leaves) and, in case of warm start, the best-case complexity of a one-step look-ahead problem ($2 m_u + 1= 9$ subproblems) is frequently approached.

The amount of information contained in the warm starts, measured as the cardinality of $ \mathscr V_\tau^0$, is very stable both in time $\tau$ and as a function of the error standard deviation $\sigma_i$.

\begin{figure}[t]
\centering
{\includegraphics[width = .9\columnwidth]{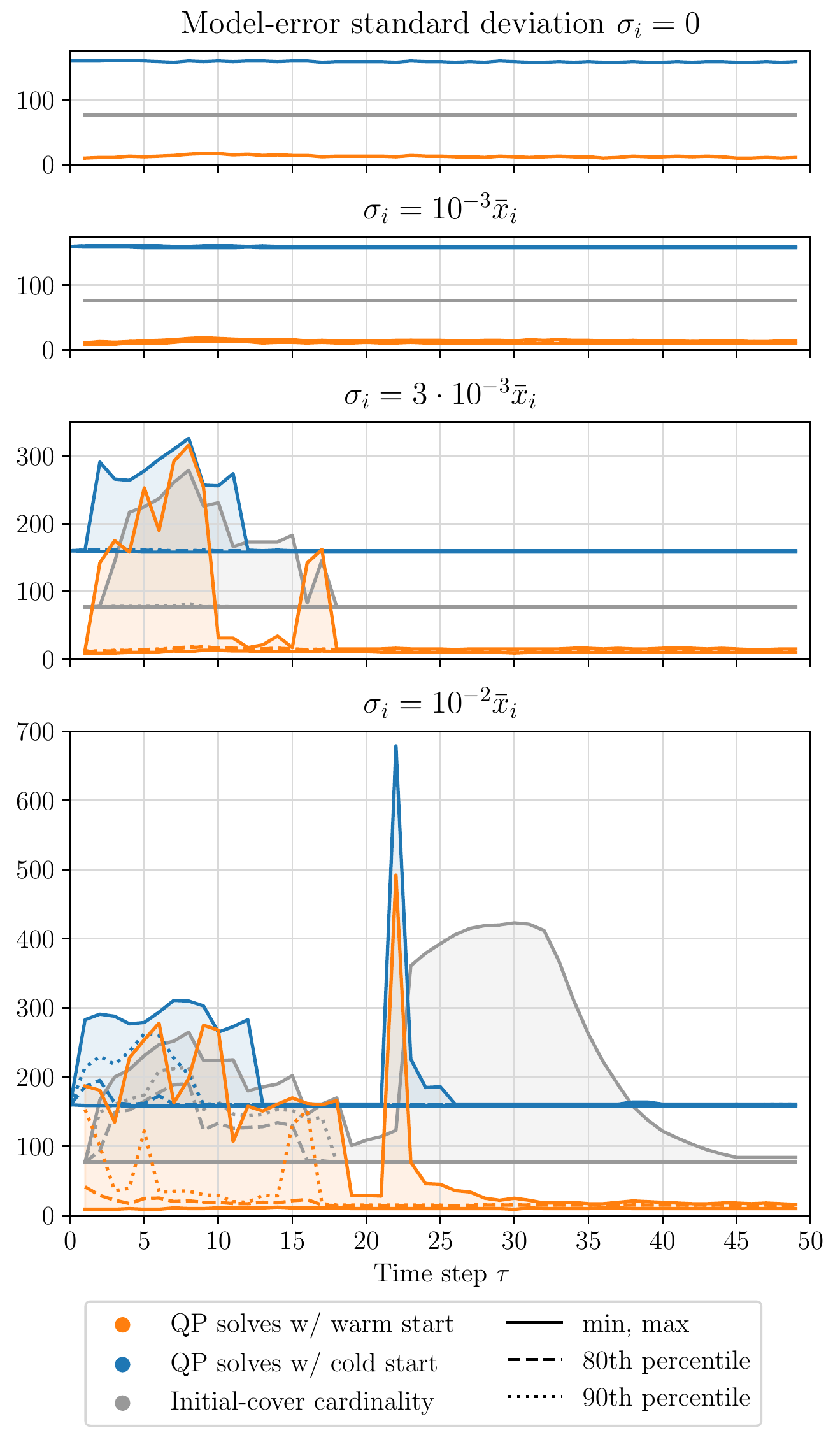}}
\caption{
Statistical analysis of the number of B\&B subproblems necessary to solve the MIQP~\eqref{eq:miqp} and regulate the cart-pole system to the origin.
Orange and blue lines: number of QP solves with warm and cold start, respectively, as functions of time and for different standard deviations of the model error $e_\tau$.
Gray lines: amount of information propagated between time steps by the warm start (represented by the cardinality of the initial cover $\mathscr V_\tau^0$) as a function of time and the error standard deviation.
Solid, dashed, and dotted lines: minimum and maximum, 80th percentile, and 90th percentile, respectively, of the above quantities over 100 feasible trial trajectories.
}
\label{fig:nodes}
\end{figure}

\subsubsection{Computation Times}

In Figure~\ref{fig:times} we illustrate the computation times of the statistical analysis.\footnote{Computations are performed on a machine with processor 2.4 GHz 8-Core Intel Core i9 and memory 64 GB 2667 MHz DDR4.}
We compare three alternatives to solve problem~\eqref{eq:miqp}: the proposed warm-start algorithm, its cold-started counterpart, and the state-of-the-art solver \texttt{Gurobi~9.0.1}.
Together with these, we report the time delay in the solution of~\eqref{eq:miqp} due to the construction of the warm start.

For our implementation of B\&B, both warm and cold started, in Figure~\ref{fig:times} we report only the time spent solving QPs (retrieved via the \texttt{Runtime} attribute of the \texttt{Gurobi} QP model).
This because almost the totality of the remaining time is spent within the \texttt{gurobipy} interface, doing array manipulations in \texttt{numpy}, or within \texttt{python} list comprehensions.
Currently, QP solves take around 15\% of the overall B\&B function-call time.
However, with a more mature implementation (e.g., in \texttt{C++}) we expect to reduce this overhead by two orders of magnitude, making it one order of magnitude smaller than the QP solve times.
For the warm-start construction times, we separate computations that can be done in the background of the time step $h$ (such as the assembly of the initial cover), and computations that require the knowledge of the current state $x_\tau$.
In  Figure~\ref{fig:times} we report only the second: despite the unoptimized \texttt{python} implementation, in this analysis, the first take just a few milliseconds (median $3$~ms, maximum $24$~ms), which is smaller than the MIQP solve times and, hence, of any reasonable sampling time $h$.

We let \texttt{Gurobi} run with default options and, to maximize its performance, we use the shifted optimal solution from the previous time step as initial guess for the binary variables (this is used by the \texttt{Gurobi} heuristics to attempt to build an initial binary assignment).
More precisely, we set the initial guess $v\at{t}{\tau+1} = v\at{t+1}{\tau}^*$ for $t = 0, \ldots, T-2$ and $v\at{T-1}{\tau+1} = 0$ (where $0$ is the equilibrium value of the binary inputs).

The comparison between warm and cold start is in line with the one above: the great majority of the times the warm-started B\&B is an order of magnitude faster and, even in the worst case, it is not slower than the cold-started one.
When warm started, our implementation frequently approaches the solution times of \texttt{Gurobi}, which is widely recognised to be the baseline solver for hybrid MPC~\cite{bemporad2015solving,naik2017embedded,stellato2018embedded,bemporad2018numerically,hespanhol2019structure}.
This is a promising result, especially considering that our implementation is single threaded, whereas we let \texttt{Gurobi} run on 16 threads.
Additionally, \texttt{Gurobi} makes a heavy use of presolve techniques and heuristics that our implementation does not feature.\footnote{
With a custom QP solver, further computation savings could be brought by linear-algebra routines specialized for the MPC sparsity pattern~\cite{rao1998application,frison2013efficient}.
}
Finally, we highlight that the delays in the MIQP solves due to the construction of the warm starts are negligible since they require, in the worst case, $10^{-3}$~s.

\begin{figure}[t]
\centering
{\includegraphics[width = 1.\columnwidth]{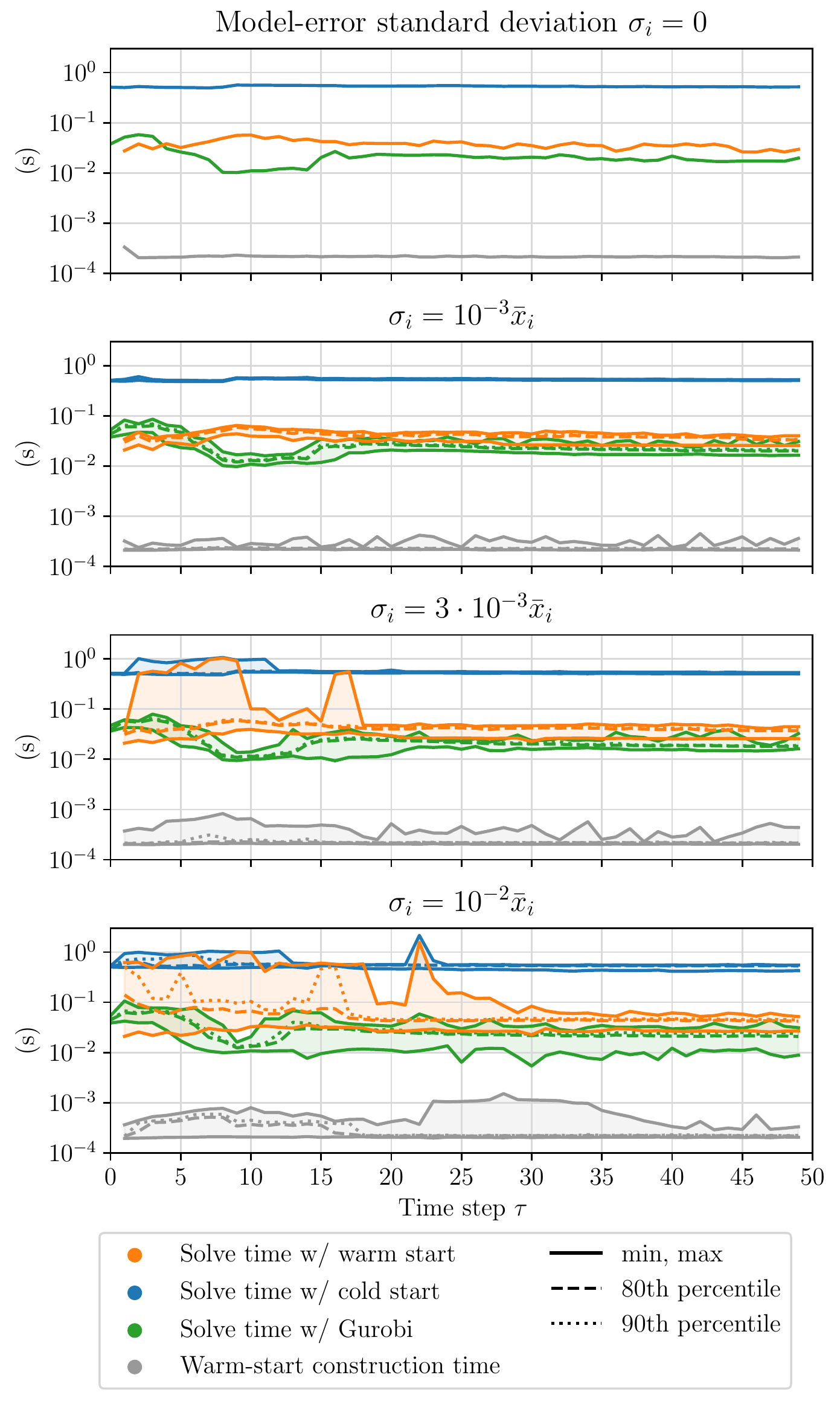}}
\caption{
Statistical analysis of the solution times for problem~\eqref{eq:miqp}.
Orange, blue, and green lines: solve times with warm start, cold start, and \texttt{Gurobi}, respectively, as functions of time and for different standard deviations of the model error $e_\tau$.
(\texttt{Gurobi} runs with default options and is allowed to use the shifted solution from the previous time step as initial guess.)
Gray lines: delay in the solution of the MIQPs due to the construction of the warm start as a function of time and the error standard deviation.
Solid, dashed, and dotted lines: minimum and maximum, 80th percentile, and 90th percentile, respectively, of the above quantities over 100 feasible trial trajectories.
}
\label{fig:times}
\end{figure}

%% file: sections/conclusions.tex
The solution of a hybrid MPC problem via B\&B generally amounts to a very large number of convex optimizations.
In this paper we have shown how, leveraging the receding-horizon structure of the problem, computations performed at one time step can be efficiently reused to warm start subsequent solves, greatly reducing the number of B\&B subproblems.

A warm start for a B\&B solver should include three elements: a collection of sets which covers the search space, a lower bound on the problem objective in each of these sets, and an upper bound on the problem optimal value.
We have shown how the first can be generated by a simple shift in time of the B\&B frontier from the previous solve.
For the second we have used duality: dual solutions of the B\&B frontier, if properly shifted, lead to lower bounds for the leaves of the new problem, even in presence of arbitrary model errors.
Finally, we have illustrated how standard persistent-feasibility arguments can be applied to synthesize the third element.
All these three ingredients take a negligible time to be computed.

We have thoroughly analyzed the tightness of the bounds we derived, revealing a connection between them and the decrease rate of the MPC cost to go.
This has led to the observation that, as the problem horizon grows to infinity, the complexity of the hybrid MPC problem tends to that of a one-step look-ahead problem.
In this case, the warm-started B\&B needs to reoptimize only the final stage of the control problem.

Theoretical results have been validated by a thorough statistical analysis.
The latter has demonstrated that our method greatly outperforms the standard approach of solving each optimization problem from scratch.

%% file: sections/extensions.tex
We collect here extensions and additional applications of the proposed algorithm.
In Appendix~\ref{sec:binary_mld} we extend the results to the case in which the MLD system to be controlled has binary states.
Appendix~\ref{sec:time_varying_mld} deals with time-varying MLD systems.
Finally, in Appendix~\ref{sec:variable_horizon}, we analyze the case in which the prediction horizon $T$ is included among the decision variables of the MPC problem~\eqref{eq:miqp}.

\subsection{MLD System with Binary States}
\label{sec:binary_mld}
\input{sections/binary_mld.tex}

\subsection{Time-Varying MLD System}
\label{sec:time_varying_mld}
\input{sections/time_varying_mld.tex}

\subsection{Variable-Horizon MPC}
\label{sec:variable_horizon}
\input{sections/variable_horizon.tex}

%% file: sections/binary_mld.tex
As discussed in Section~\ref{sec:mld}, auxiliary inputs can be used to constrain state components to assume binary values.
However, this approach might be suboptimal from the viewpoint of computational efficiency: in this appendix, we show how binary states can be explicitly included in the analysis.

We consider the state vector $x_\tau \in \mathbb R^{n_x} \times \{0,1\}^{m_x}$, and we denote by $y_\tau \in \{0,1\}^{m_x}$ its binary entries.
We define the selection matrix $Y$ so that $y_\tau = Y x_\tau$.
The error vector $e_\tau$ takes now values in $\mathbb R^{n_x} \times \{-1,0,1\}^{m_x}$.
In these settings, problem~\eqref{eq:miqp} must include the additional constraint $Y x\at{t}{\tau} \in \{0,1\}^{m_x}$ for $t = 0, \ldots,T$, and B\&B needs to find a cover of $\{0,1\}^{(T + 1) m_x + T m_u}$.
In the convex relaxation of $\mathbf P_\tau$, we have the additional constraints $\ubar y\at{t}{\tau} \leq Y x\at{t}{\tau} \leq \bar y\at{t}{\tau}$ for $t=0, \ldots, T$, where $\ubar y\at{t}{\tau}, \bar y\at{t}{\tau} \in \{0,1\}^{m_x}$ and $\ubar y\at{t}{\tau} \leq \bar y\at{t}{\tau}$.
Let $\ubar \xi\at{t}{\tau}$ and $\bar \xi\at{t}{\tau}$ be the nonnegative multipliers associated with these constraints.
The dual objective~\eqref{eq:dual_objective} must now include the additional linear term $\sum_{t=0}^T  (\ubar y\at{t}{\tau}' \ubar \xi\at{t}{\tau} - \bar y\at{t}{\tau}' \bar \xi\at{t}{\tau})$.
Similarly, the terms $Y' (\bar \xi\at{t}{\tau} - \ubar \xi\at{t}{\tau})$ and $Y' (\bar \xi\at{T}{\tau} - \ubar \xi\at{T}{\tau})$ must be added to the left-hand sides of~\eqref{eq:dual_x_t} and~\eqref{eq:dual_x_T}, respectively.
The logic behind the shifting procedure from Section~\ref{sec:initial_cover} is the same: the first step requires the additional check $\ubar y\at{0}{0} \leq y_0 \leq \bar y\at{0}{0}$, the second generates the additional bounds $(\ubar y\at{1}{0}, \ldots, \ubar y\at{T}{0} , 0, \ldots, 0)$, $(\bar y\at{1}{0}, \ldots, \bar y\at{T}{0}, 1, \ldots, 1)$.
In Lemma~\ref{lemma:propagation_dual_feasible_solution}, we define $(\ubar \xi\at{t}{1}, \bar \xi\at{t}{1}) := (\ubar \xi\at{t+1}{0}, \bar \xi\at{t+1}{0})$ for $t=0, \ldots, T-1$, and $(\ubar \xi\at{T}{1}, \bar \xi\at{T}{1}) := 0$.
In Theorem~\ref{th:propagation_bounds}, we add to $\pi_3$ the nonnegative term $(y_0 - \ubar y\at{0}{0})' \ubar \xi\at{0}{0} + (\bar y\at{0}{0} - y_0)' \bar \xi\at{0}{0}$, whereas Corollary~\ref{cor:propagation_certificate_infeasibility} remains unchanged.
In Section~\ref{sec:upper_bound}, we just need the additional condition $Y x \in \{0,1\}^{m_x}$ for control invariance.
The asymptotic analysis from Section~\ref{sec:asymptotic_analysis} and the extensions in Sections~\ref{sec:terminal} do not require any modification.

%% file: sections/time_varying_mld.tex
All the results presented in this paper can be immediately generalized to the case of a time-varying MLD system
\begin{align}
x_{\tau+1} = A_\tau x_\tau + B_\tau u_\tau + e_\tau, \quad (x_\tau, u_\tau) \in \mathcal D_\tau.
\end{align}
Dynamics of this kind appear, e.g., in trajectory tracking or local stabilization of limit cycles for hybrid nonlinear systems.
These are two common problems in robotics, where state-of-the-art methods cannot reason yet about online modifications of the preplanned switching sequence~\cite{sanfelice2007hybrid,manchester2011stable,manchester2011transverse,farshidian2017sequential}.

In this case, the dynamics~\eqref{eq:miqp_dynamics} becomes $x\at{t+1}{\tau} = A_{\tau+t} x\at{t}{\tau} + B_{\tau+t} u\at{t}{\tau}$, the constraint~\eqref{eq:miqp_constraints} reads $(x\at{t}{\tau}, u\at{t}{\tau}) \in \mathcal D_{\tau+t}$, and the weight matrices $Q_{\tau+t}$ and $R_{\tau+t}$ can vary with the \emph{absolute time} $\tau+t$.
Note that this time dependency is easier to handle than the one discussed in Section~\ref{sec:terminal}.
There, problem data depend on the relative time $t$ and they can disagree after a shift of the MPC time window (e.g., the matrix $Q_{t+1}$ in problem $\mathbf P_\tau$ might be different from $Q_t$ in $\mathbf P_{\tau+1}$).
Here, problem data still match after a window shift, and procedures like the one from Lemma~\ref{lemma:propagation_dual_feasible_solution} do not break.

The dual problem~\eqref{eq:dual} does not change structure, it only requires a suitable modification of the subscripts of the matrices in it.
The shifting procedure in Section~\ref{sec:initial_cover} does not need any adjustment.
The results from Section~\ref{sec:lower_bounds} are also still valid, provided that we add the subscript $0$ to the matrices $Q$, $R$, $F$, $G$, and $h$ in the statement of Theorem~\ref{th:propagation_bounds}.
Also the persistent-feasibility argument from Section~\ref{sec:upper_bound} extends to the time-varying case: we now have a sequence of control-invariant sets $\mathcal X_\tau$ and, for all $x$ in $\mathcal X_\tau$, there must exists a $u \in \mathbb R^{n_u} \times \{0,1\}^{m_u}$ such that $(x,u) \in \mathcal D_{\tau+T}$ and $A_{\tau+T} x + B_{\tau+T} u \in \mathcal X_{\tau+1}$.
The asymptotic considerations from Section~\ref{sec:asymptotic_analysis} need only a couple of adjustments: $Q$ and $R$ in Lemma~\ref{lemma:lower_bound_theta1} must be substituted with $Q_0$ and $R_0$, and the invariance argument in Theorem~\ref{th:recursive_cover} must be revised as just shown with persistent feasibility.

The extension of the results from Section~\ref{sec:terminal} is slightly more involved.
The weight matrices and the constraint sets depend now on $\tau$ and $t$ independently: we use the notation $Q\at{t}{\tau}$, $R\at{t}{\tau}$, $\mathcal D\at{t}{\tau}$ for the data of problem $\mathbf P_\tau$ at time $t$.
Note that, e.g., a terminal penalty implies $Q\at{T}{\tau} \neq Q\at{T-1}{\tau+1}$.
Assumption~\ref{ass:row_space_cost} must now require that the row space of $Q\at{t}{\tau+1}$ and $R\at{t}{\tau+1}$ contains the one of $Q\at{t+1}{\tau}$ and $R\at{t+1}{\tau}$, respectively.
Then, the generalization presented in Section~\ref{sec:assumption_cost} also applies to the time-varying case if, e.g., instead of the matrices $Q_t$ and $Q_{t+1}$, we consider $Q\at{t}{1}$ and $Q\at{t+1}{0}$.
Analogous changes are required for Assumption~\ref{ass:row_space_constraints} and the results from Section~\ref{sec:assumption_constraints}.

%% file: sections/variable_horizon.tex
In many applications, it is desirable to include the time horizon $T$ among the decision variables of the MPC problem.
Besides avoiding the tricky compromise of fixing a value for $T$, this guarantees persistent feasibility and minimizes the discrepancy between open- and closed-loop trajectories~\cite{michalska1993robust,scokaert1998min}.
Additionally, it extends the scope of MPC beyond regulation problems~\cite{richards2006robust}.

A common problem statement for \emph{variable-horizon} MPC is to find a control sequence that drives the system state to a target set, despite disturbances and minimizing a weighted sum of the reach time and the control effort~\cite[Section~2.4]{richards2006robust,shekhar2012variable}.
In~\cite{richards2006robust} this problem has been transcribed in mixed-integer form by parameterizing the reach time with binary variables $b\at{t}{\tau}$ ($b\at{t}{\tau}=1$ when the target set is reached, and $b\at{t}{\tau}=0$ otherwise).
Because of coupling constraints between the binaries of different time steps (e.g., $\sum_{t=0}^T b\at{t}{\tau} = 1$), the formulation in~\cite{richards2006robust} does not have the form of an optimal control problem of MLD systems.
However, equivalent binary parameterizations that enjoy this property can be easily found, resulting in a problem of the form we considered in Section~\ref{sec:terminal} and allowing the use of the proposed warm-start technique.

%% file: sections/dual_qp.tex
In this appendix we derive the dual $\mathbf D  (\mathcal V)$  of the convex relaxation $\mathbf P (\mathcal V)$ of problem~\eqref{eq:miqp}, with $\mathcal V$ defined in~\eqref{eq:interval_tau}.
We describe this derivation since the nonstrict convexity of $\mathbf P (\mathcal V)$ requires some special care.

We start by introducing the auxiliary primal variables
\begin{subequations}
\label{eq:auxiliary_primal}
\begin{align}
& z\at{t}{\tau} := Q x\at{t}{\tau}, && t = 0, \ldots, T,
\\
& w\at{t}{\tau} := R u\at{t}{\tau}, && t=0, \ldots, T-1.
\end{align}
\end{subequations}
After substituting these in the primal objective~\eqref{eq:miqp_objective}, we define the Lagrangian function
\begin{align}
\nonumber
l := & \
\sum_{t = 0}^{T}
[| z\at{t}{\tau} |^2 + \rho\at{t}{\tau}'(Q x\at{t}{\tau} - z\at{t}{\tau})]
\\  \nonumber &
+
\sum_{t = 0}^{T-1}
[| w\at{t}{\tau} |^2 + \sigma\at{t}{\tau}'(R u\at{t}{\tau} - w\at{t}{\tau})]
+ \lambda\at{0}{\tau}' (x\at{0}{\tau} - x_\tau)
\\  \nonumber &
+
\sum_{t = 0}^{T-1}
%\Big\{
\lambda\at{t+1}{\tau}' (x\at{t+1}{\tau} - A x\at{t}{\tau} - B u\at{t}{\tau})
\\  \nonumber &
+
\sum_{t = 0}^{T-1}
\mu\at{t}{\tau}' (F x\at{t}{\tau} + G u\at{t}{\tau} - h)
\\  \label{eq:lagrangian} &
+
\sum_{t = 0}^{T-1}
[
\ubar \nu\at{t}{\tau}' (\ubar v\at{t}{\tau} - V u\at{t}{\tau})
+
\bar \nu\at{t}{\tau}' (V u\at{t}{\tau} - \bar v\at{t}{\tau})
],
%\Big\},
\end{align}
with $\{\lambda\at{t}{\tau}, \rho\at{t}{\tau}\}_{t=0}^T$ and $\{\mu\at{t}{\tau}, \ubar \nu\at{t}{\tau}, \bar \nu\at{t}{\tau}, \sigma\at{t}{\tau}\}_{t=0}^{T-1}$ Lagrange multipliers of appropriate dimensions.
For any fixed value of the multipliers such that the nonnegativity condition~\eqref{eq:dual_nonneg} holds, the infimum of the Lagrangian with respect to the primal variables yields a lower bound on the optimal value $\theta (\mathcal V)$.
We seek the multipliers for which this lower bound is maximum.

For the outer maximization to be feasible (i.e., have an optimal value greater than $-\infty$),
we must require the inner minimization to be bounded.
Since the Lagrangian is a convex quadratic function of the primal variables, its infimum, if finite, verifies the stationarity conditions
$\nabla_{x\at{t}{\tau}} l = 0$ (corresponding to~\eqref{eq:dual_x_t} and~\eqref{eq:dual_x_T}),
$\nabla_{u\at{t}{\tau}} l = 0$ (corresponding to~\eqref{eq:dual_u_t}), and
\begin{subequations}
\label{eq:stationarity_zwt}
\begin{align}
& \nabla_{z\at{t}{\tau}} l =  2 z\at{t}{\tau} - \rho\at{t}{\tau} = 0, && t = 0, \ldots, T,
\\
& \nabla_{w\at{t}{\tau}} l =  2 w\at{t}{\tau} - \sigma\at{t}{\tau} = 0, && t = 0, \ldots, T-1.
\end{align}
\end{subequations}
Substituting the stationarity conditions in the Lagrangian~\eqref{eq:lagrangian}, we obtain its minimum value~\eqref{eq:dual_objective}.
The dual problem $\mathbf D  (\mathcal V)$ consists then in the maximization of~\eqref{eq:dual_objective}, subject to the stationarity conditions and the nonnegativity of the multipliers $\{\mu\at{t}{\tau}, \ubar \nu\at{t}{\tau}, \bar \nu\at{t}{\tau}\}_{t=0}^{T-1}$.
Conditions~\eqref{eq:stationarity_zwt} are removed from the dual problem because they are redundant.

%% file: sections/proof_theorem.tex
In this appendix we derive the lower bound~\eqref{eq:lower_bound_nominal}.
Given a feasible solution for $\mathbf D_0 (\mathcal V_0)$ we define a set of feasible multipliers for $\mathbf D_1 (\mathcal V_1)$ as in Lemma~\ref{lemma:propagation_dual_feasible_solution}.
Substituting these into the objective~\eqref{eq:dual_objective} of the latter problem, we get the lower bound
\begin{multline}
\ubar \theta_1 (\mathcal V_1) :=
-
\sum_{t = 0}^{T-1}
| \rho\at{t+1}{0} / 2 |^2
-
\sum_{t = 0}^{T-2}
| \sigma\at{t+1}{0} / 2 |^2
-
x_1' \lambda\at{1}{0}
\\
-
\sum_{t = 0}^{T-2}
(
h' \mu\at{t+1}{0} 
+
\bar v\at{t+1}{0}' \bar \nu\at{t+1}{0} 
-
\ubar v\at{t+1}{0}' \ubar \nu\at{t+1}{0} 
).
\end{multline}
The cost of the candidate solution can be restated as $\ubar \theta_1 (\mathcal V_1) = \ubar \theta_0 (\mathcal V_0) + \sum_{i=1}^3 \omega_i$, where
\begin{subequations}
\begin{align}
\omega_1 := \ &
x_0' \lambda\at{0}{0}
- x_1' \lambda\at{1}{0},
\\
\omega_2 := \ &
h' \mu\at{0}{0}
+ \bar v\at{0}{0}' \bar \nu\at{0}{0}
- \ubar v\at{0}{0}' \ubar \nu\at{0}{0},
\\
\omega_3 := \ &
| \rho\at{0}{0} / 2|^2
+
| \sigma\at{0}{0} / 2|^2.
\end{align}
\end{subequations}
Enforcing the dynamics, we get
\begin{align}
\omega_1 = x_0' \lambda\at{0}{0} - (A x_0 + B u_0 + e_0)' \lambda\at{1}{0},
\end{align}
and using~\eqref{eq:dual_x_t} and~\eqref{eq:dual_u_t} for $t=\tau=0$, we have
\begin{multline}
\omega_1 = 
- x_0' (Q' \rho\at{0}{0} + F' \mu\at{0}{0} )
\\
- u_0' [
R'  \sigma\at{0}{0} + G'  \mu\at{0}{0} + V' (\bar \nu\at{0}{0} - \ubar \nu\at{0}{0})
] +   \pi_4.
\end{multline} 
Adding $\omega_2$, we obtain
\begin{align}
\omega_1 + \omega_2 =
- x_0'  Q' \rho\at{0}{0} 
- u_0' R' \sigma\at{0}{0}
+ \pi_3 + \pi_4.
\end{align}
Finally, we add $\omega_3$:
\begin{multline}
\sum_{i=1}^3 \omega_i =
| \rho\at{0}{0} / 2|^2 
-
x_0' Q' \rho\at{0}{0}
+
| \sigma\at{0}{0} /2 |^2 
-
u_0' R' \sigma\at{0}{0}
\\
+ \pi_3 + \pi_4.
\end{multline}
Using the identities
\begin{subequations}
\begin{align}
| \rho\at{0}{0} / 2|^2  - x_0' Q' \rho\at{0}{0} &= | \rho\at{0}{0} / 2 - Q x_0 |^2 - | Q x_0 |^2, \\
| \sigma\at{0}{0} / 2|^2  - u_0' R' \sigma\at{0}{0} &= | \sigma\at{0}{0} / 2 - R u_0 |^2 - | R_0 u_0 |^2,
\end{align}
\end{subequations}
and recalling the definition of $\pi_1$ and $\pi_2$, we obtain
$
\sum_{i=1}^3 \omega_i = \sum_{i=1}^4 \pi_i,
$
and hence~\eqref{eq:lower_bound_nominal}.

%% file: sections/acknowledgment.tex
This research was supported by the Grass Instruments Company and the Department of the Navy, Office of Naval Research, Award No. N00014-18-1-2210. Any opinions, findings, and conclusions or recommendations expressed in this material are those of the authors and do not necessarily reflect the views of the Office of Naval Research.

The authors thank Twan Koolen for the many helpful comments on the original manuscript.

%% file: sections/tobia_bio.tex
graduated cum laude in Mechanical Engineering from the University of Pisa in 2015.
From 2015 to 2017 he was Ph.D. student at the Research Center “E. Piaggio”, University of Pisa, and the Istituto Italiano di Tecnologia (IIT).
Since 2017 he is at the Computer Science and Artificial Intelligence Laboratory (CSAIL), MIT, to continue his Ph.D. studies.
His main research interests are robotics, control theory, and numerical optimization.

%% file: sections/russ_bio.tex
is the Toyota Professor of Electrical Engineering and Computer Science, Aeronautics and Astronautics, and Mechanical Engineering at MIT, the Director of the Center for Robotics at CSAIL, and the leader of Team MIT's entry in the DARPA Robotics Challenge.
Russ is also the Vice President of Robotics Research at the Toyota Research Institute.
He is a recipient of the NSF CAREER Award, the MIT Jerome Saltzer Award for undergraduate teaching, the DARPA Young Faculty Award in Mathematics, the 2012 Ruth and Joel Spira Teaching Award, and was named a Microsoft Research New Faculty Fellow.
Russ received his B.S.E. in Computer Engineering from the University of Michigan, Ann Arbor, in 1999, and his Ph.D. in Electrical Engineering and Computer Science from MIT in 2004, working with Sebastian Seung.
After graduation, he joined the MIT Brain and Cognitive Sciences Department as a Postdoctoral Associate.
During his education, he has also spent time at Microsoft, Microsoft Research, and the Santa Fe Institute.